\documentclass[10pt,journal,compsoc]{IEEEtran}
\usepackage{cite}
\usepackage{amsmath,amssymb,amsfonts, amsthm}
\usepackage{graphicx}
\usepackage{textcomp}
\usepackage{xcolor}

\definecolor{Xiang}{rgb}{1,0,0}
\definecolor{Niranjan}{rgb}{0,0,1}
\definecolor{New}{rgb}{0,0.5,0}

\usepackage{esvect} %vector
\usepackage{url}
\newtheorem{definition}{Definition}
\newtheorem{example}{Example}
\usepackage{subfigure}
\usepackage{tabu}
\usepackage{algpseudocode}
\usepackage{balance}
\usepackage[linesnumbered, ruled, vlined]{algorithm2e}
\usepackage[export]{adjustbox}

\newcommand{\nop}[1]{}
\newtheorem{lemma}[]{Lemma}
\newtheorem{theorem}[]{Theorem}

\ifCLASSINFOpdf
\else
\fi
\begin{document}
\title{Top-$k$ Community Similarity Search Over Large-Scale Road Networks\\ (Technical Report)}

\author{Niranjan Rai and %,~\IEEEmembership{Member,~IEEE,}
        Xiang Lian %,~\IEEEmembership{}
        % <-this % stops a space
\IEEEcompsocitemizethanks{\IEEEcompsocthanksitem Niranjan Rai and Xiang Lian are with the Department
of Computer Science, Kent State University.\protect\\
E-mail: \{nrai, xlian\}@kent.edu}% <-this % stops an unwanted space
\thanks{Manuscript received May, 2021.}}


\markboth{IEEE Transactions on Knowledge and Data Engineering}%
{Rai and Lian: Top-$k$ Community Similarity Search Over Large-Scale Road Networks (Technical Report)}



\IEEEtitleabstractindextext{%
\begin{abstract}
With the urbanization and development of infrastructure, the community search over 
road networks has become increasingly important in many real applications such as 
urban/city planning, social study on local communities, and community recommendations by 
real estate agencies. In this paper, we propose a novel problem, namely \textit{top-$k$ 
community similarity search} ($Top\text{-}kCS^2$) over road networks, which efficiently and effectively 
obtains $k$ spatial communities that are the most similar to a given query community 
in road-network graphs. In order to efficiently and effectively tackle the $Top\text{-}kCS^2$ problem,
in this paper, we will design an effective similarity measure between spatial communities, and propose a framework for retrieving $Top\text{-}kCS^2$ query answers, which integrates offline pre-processing and online computation phases. Moreover, we also consider a variant, namely \textit{continuous top-$k$ community similarity search} ($CTop\text{-}kCS^2$), where the query community continuously moves along a query line segment. We develop an efficient algorithm to split query line segment into intervals, incrementally obtain similar candidate communities for each interval, and refine actual $CTop\text{-}kCS^2$ query answers. Extensive experiments have been conducted on real and synthetic data sets to confirm the efficiency and effectiveness of our proposed $Top\text{-}kCS^2$ and
$CTop\text{-}kCS^2$ approaches under various parameter settings. 
\end{abstract}

}

\maketitle

\IEEEdisplaynontitleabstractindextext

% \fi
\IEEEpeerreviewmaketitle

\IEEEraisesectionheading{\section{Introduction}\label{sec:introduction}}

%{\color{Xiang} \bf (Niranjan, please add some citations below. Please list more applications like biological networks for community search/detection.)}

%In introduction, we should motivate in the examples why we consider closer distance of the community.

Recently, the community search/detection over graphs has received much attention in many real-world applications such as social network analysis \cite{fortunato2010community, newman2004finding, xu2012model, liu2009topic,nallapati2008joint, zhou2009graph, sozio2010community, cui2014local, cui2013online, li2015influential,huang2015approximate}, online marketing and advertising over geo-social networks \cite{yang2012socio, Li12, Yuan16, fang2017effective, GPSSN, chen2018maximum}, and many others. While prior works on the community search/detection \cite{Fortunato10, Newman04, sozio2010community, cui2014local, cui2013online, LiQYM15, huang2015approximate, FangCLH16} usually considered \textit{user communities} with strong social/spatial relationships in (geo-)social networks, in this paper, we will study a novel problem of retrieving top-$k$ spatial communities on road-network graphs, which are quite useful and important for urban/city planning or community recommendations by real estate agencies.

We have the following motivation examples.

%{\color{blue} 
\begin{example}
{\bf (Data Visualization via Lenses on Road Networks)} In real applications such as urban/city planning, social study, or transportation systems, data analysts often utilize geospatial visualization tools such as interactive lens \cite{article} and identify/analyze those communities with neighborhood similar to a target (query) community. Figure \ref{fig:roadnetwork} illustrates a map of road networks in a visualization system, on which a lens (i.e., a circle with radius $r$) is specified by a user (e.g., geologist or data analyst). In this case, the road-network subgraph within the lens can be considered as a query community, and the user may be interested in finding/analyzing other communities on road networks 
in a city (or a county) that are similar to the query community within the lens (with similar road-network structures and points of interest). \qquad $\blacksquare$
\end{example}

\begin{figure}[t!]
    \centering
    \includegraphics[width=170pt]{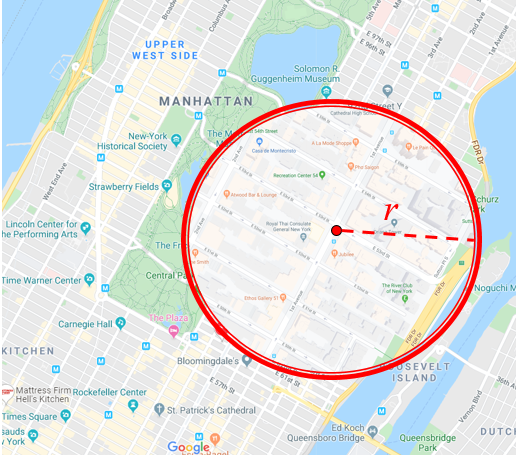}
    \caption{\small An example of lens on a road-network graph $G$.}
    \label{fig:lens}
\end{figure}

\begin{example}
{\bf (Neighborhood Recommendations to a Moving Family)} Due to various reasons such as job changes, finding new school districts, and so on, people may often move from one place to another. Assume that a family has to move to a new neighborhood due to a job change and prefer to live in a new neighborhood which is similar to their old one and stays spatially close to their new work location. For example, similar to the old neighborhood, the new neighborhood should have safe roads, easy commute routes, and/or desirable \textit{points of interest} (POIs) such as movie theatres or restaurants. Note that, many studies \cite{roundabout} have shown that the structure of road networks plays a vital role in the number of accidents. For example, the number of accidents significantly decreases, when road intersection points are replaced by round-abouts \cite{roundabout}. Moreover, family members (e.g., children) are usually interested in some POIs (e.g., theatres or parks).

Therefore, in this scenario, a realtor may want to find and recommend $k$ spatial communities (subgraphs) over road-network graphs that are most similar to a query community and spatially close to the family's new location (e.g., new working place or school district), in terms of graph structure, POI similarity, and spatial closeness. The returned spatial communities are useful for recommending candidate communities to the family to move in. $\blacksquare$

\end{example}

Inspired by the examples above, in this paper, we will formulate and tackle a novel problem, namely \textit{top-$k$ community similarity search} ($Top\text{-}kCS^2$), which efficiently and effectively obtains top-$k$ spatial communities that are similar and spatially close to a given query community over road networks. 

Note that, efficient and effective answering of the $Top\text{-}kCS^2$ query is rather challenging. A straightforward method to process the $Top\text{-}kCS^2$ query is to enumerate all possible communities (subgraphs) in road networks, compute the similarity/distance between each community and the query community, and return $k$ communities with higher similarities  than a given threshold, $\theta$ and small distances. However, this straightforward method is not very efficient, due to the large number of candidate communities to refine on road networks. What is more, it is not very trivial how to accurately define the similarity between two communities that captures their graph structural and POI similarities.

To our best knowledge, prior works (e.g., community search in (geo-)social networks) did not consider finding similar/close spatial communities on large-scale road-network graphs. Therefore, previous techniques cannot be directly applied to solve our $Top\text{-}kCS^2$ problem. In order to tackle the challenges of processing $Top\text{-}kCS^2$ queries, in this paper, we will propose a novel metric to measure the similarity between two spatial communities in road-network graphs, design effective pruning strategies (w.r.t. similarity and distance) to reduce the $Top\text{-}kCS^2$ search space, as well as an effective indexing mechanism, and develop an efficient $Top\text{-}kCS^2$ processing algorithm via index that integrates our proposed pruning methods.

Furthermore, we also formulate and tackle a variant, that is, continuous $Top\text{-}kCS^2$ problem (denoted as $CTop\text{-}kCS^2$), where the query community moves along a query line segment (e.g., surrounding communities from home to the working place). We propose to split the query line segment into multiple intervals (each with the same query community), and design efficient algorithms to monitor and maintain top-$k$ communities that are similar to the query community for each interval.

%It is especially due to the large number of intersections and junctions. To effectively compute the communities in such large road-networks is a challenging work. In our $Top\text{-}kCS^2$ problem, to improve the accuracy of similarity, we also consider the point of interests, $POI$, such as Hotels, Restaurants, Schools, Churches, etc., located in each edges of road-network graph. In doing so, we introduce more challenges to our solution. But, in this paper, we propose effective methods to overcome these challenges and efficiently retrieve more accurately similar communities. We propose several pruning heuristics as well as use a well-known indexing technique, aggregate R-Tree $(aR-Tree)$, which helps us quickly prune false alarms and retrieve top-$k$ communities effectively and efficiently.

In this paper, we make the following contributions.
\begin{itemize}
    \item We formally define a novel problem, namely \textit{top-$k$ community similarity search} ($Top\text{-}kCS^2$) query, over road-network graphs in Section \ref{sec:prob_def}.
    \item We present the framework for answering the $Top\text{-}kCS^2$ query in Section \ref{sec:topk_frame}.
    \item We illustrate the offline pre-processing phase of the framework in Section \ref{sec:offline_processing}.
    \item We devise effective pruning methods for $Top\text{-}kCS^2$ in Section \ref{sec:prune}.
    \item We provide the heuristics for retrieving candidate unit patterns and develop an efficient algorithm to obtain such candidate unit patterns via the index traversal in Section \ref{sec:candidate_unit}.
    \item We design efficient algorithms to answer $CTop\text{-}kCS^2$ queries on road networks in Section \ref{sec:ctopksol}.
    \item We demonstrate the efficiency and effectiveness of our proposed $Top\text{-}kCS^2$ and $CTop\text{-}kCS^2$ approaches in Section \ref{sec:experiment}.
\end{itemize}

Section \ref{sec:related_work} reviews previous works on community search/detection in (geo-)social networks. Section \ref{sec:conclusion} concludes this paper. 

\section{Problem Definition}
\label{sec:prob_def}

In this section, we formally define \textit{top-k community similarity search over road-network graphs}.

\begin{figure}[ht]
    \centering
    \includegraphics[width=190pt]{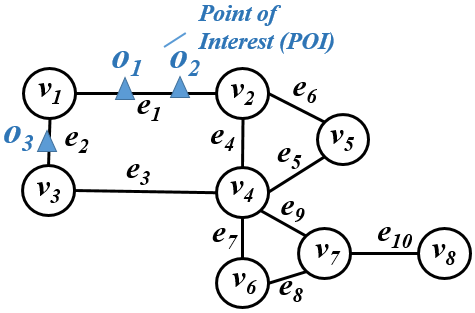}
    \caption{\small An example of a road-network graph $G$.}
    \label{fig:roadnetwork}
\end{figure}

\subsection{Road-Network Graphs}

\noindent {\bf Road Networks.} In this paper, we model road networks by a planar graph, defined as follows.

\begin{definition}
\textbf{(Road-Network Graph)} A road-network graph is a connected planar graph $G = (V(G), E(G), \Phi(G))$, where $V(G)$ and $E(G)$ are the sets of vertices and edges in graph $G$, respectively, and $\Phi(G)$ is a mapping function: $V(G) \times V(G) \to E(G)$. \label{def:graph}
\end{definition}

In Definition \ref{def:graph}, edges $e_i\in E(G)$ represent roads (line segments) in road networks $G$, where each edge $e_i$ is associated with its length $e_i.l$. Moreover, vertices $v_i \in V(G)$ correspond to intersection points of road segments.

\begin{example}
Figure \ref{fig:roadnetwork} shows an example of a road-network graph $G$, where the vertex set $V(G)$ = $\{v_1, v_2, ..., v_{8} \}$ and the edge set $E(G) = \{e_1, e_2, ..., e_{10}\}$. For example, edge $e_1$ is a road segment connecting 2 ending vertices $v_1$ and $v_2$. $\blacksquare$
\end{example}

\noindent {\bf Points of Interest.} On edges $e_i \in E(G)$ of road networks $G$, there are a number of \textit{points of interest} (POIs), such as restaurants and movie theatres, which are defined as follows:

\begin{definition}
\textbf{(Points of Interest, POI)} Given a road-network graph $G$, a \textit{point of interest} (POI), $o_j$, is a facility (object) located at $o_j.loc$ on an edge $e_i\in E(G)$.\label{def:POI}
\end{definition}

In Definition \ref{def:POI}, POIs on edge $e \in E(G)$ can be of various types, such as restaurants, shopping malls, supermarkets, cinemas, schools, churches, houses, and so on. We can represent all POIs on edge $e \in E(G)$ by a POI vector, $e.vec$, which consists of counts (frequencies) of different POI types on edge $e$. For example, assume that we only consider 4 types of POIs, restaurant, church, house, and school. If an edge $e$ contains 2 restaurants, 1 church, 3 houses, and 5 schools, then its POI vector $e.vec$ is given by $e.vec = (2, 1, 3, 5)$.

\begin{example}
\noindent As illustrated in Figure \ref{fig:roadnetwork}, there are two POI objects $o_1$ and $o_2$ on edge $e_1$, where object $o_1$ represents a house and object $o_2$ is a school. Thus, the POI vector, $e_1.vec$, of edge $e_1$ is given by $(1, 1)$.   $\qquad$  $\blacksquare$
\end{example}

\subsection{The Spatial Community in Road-Network Graphs}

Before we define the spatial community on road networks, we first discuss \textit{patterns} and \textit{unit patterns} in road-network graphs.

\vspace{1ex}\noindent {\bf Road-Network Patterns.} As shown in Figure \ref{fig:patterns}, there are many possible structural patterns in road-network graphs, which may indicate different scenarios of road-network designs. For example, in city areas, it is very likely that we have a large number of \textit{grid} structures of rectangular shape, due to the block system planning. Note that, the grid pattern has a lot of intersections and short roads, and a study \cite{roundabout} has shown that the number of accidents is higher for this pattern type than others. 

The \textit{tree} pattern is another common pattern found on road networks. This type of pattern is very common in residential areas, where houses, churches, and/or schools are usually located at the leaves of trees. 

The \textit{radial} pattern is composed of a network of roads, which radiate from a core. Such a pattern type usually indicates a business area. 

The \textit{mesh} pattern is usually the results of unplanned road networks, where there are a lot of structures of different pattern types.

\begin{figure}[t!]
    \centering \vspace{-1ex}
    \includegraphics[width=230pt]{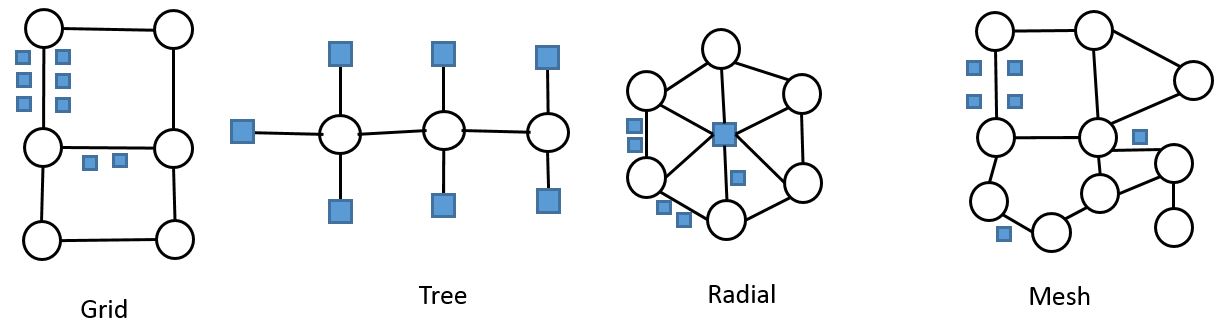}
    \caption{\small An example of patterns in road-network graphs.}
\end{figure}

\vspace{1ex}\noindent {\bf Unit Patterns in Road-Network Graphs.} Figure \ref{fig:structures} illustrates several basic patterns, called \textit{unit patterns}, on road networks, which include \textit{edge}, \textit{delta}, \textit{rectangle}, \textit{pentagon}, \textit{hexagon}, and so on. In particular, the \textit{edge} unit pattern is an edge (a road segment, but not in a circle), which can be a branch in the tree pattern or a dead end in residential areas. Similarly, the \textit{delta} unit pattern contains 3 vertices, forming a circular triangular structure. 

In this paper, we consider road networks as a planar graph. Thus, we can divide this planar road-network graph into non-overlapping unit patterns of different types. Intuitively, unit patterns such as rectangles correspond to blocks on road networks.

\begin{figure}[t!]
    \centering
    \includegraphics[width=230pt]{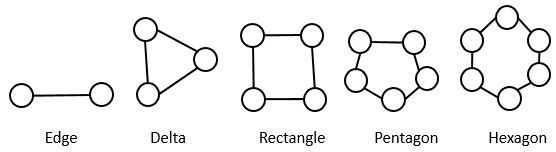}
    \caption{\small An example of unit patterns in road-network graphs.}
    \label{fig:structures}
\end{figure}

\vspace{1ex}\noindent {\bf Spatial Community.} Next, we give the definition of a spatial community in road-network graphs.

\begin{definition}

\textbf{(Spatial Community in a Road-Network Graph)}. Given a road-network graph $G$, a center vertex $v_c \in V(G)$, and a radius $r$, a \textit{spatial community}, $C_l$, is a subgraph of $G$ (i.e., $C_l \subseteq G$), such that:
\begin{enumerate}
    \item subgraph $C_l$ is connected, and;
    \item all unit patterns $c$ in $C_l$ have the minimum distances to vertex $v_c$ less than or equal to $r$ (i.e., $mindist(c, v_c) \leq r$),
\end{enumerate}
where $mindist(c, v_c)$ computes the minimum Euclidean distance from vertex $v_c$ to unit pattern $c$.
\end{definition}

Intuitively, a spatial community is a subgraph of road networks $G$, whose unit patterns (i.e., blocks) intersect with a circle centered at vertex $v_c$ and with radius $r$. 

Note that, in this paper, we assume that radius $r$ is a pre-defined system parameter, which can be specified by the system (e.g., the radius of lens in a visualization system) or tuned/inferred from historical users' preferences (e.g., 10 miles within users' driving distances).

\begin{example}
In the example of Figure \ref{subfig:com_example}, assume that we have a center vertex $v_4$, and a radius $r$. Then, a spatial community, $C_4$, centered at vertex $v_4$ and with radius $r$, is given by a subgraph with vertices $\{v_1, v_2, v_3, v_4, v_5,v_6, v_7,v_8\}$ and edges $e_1 \sim e_{10}$. Note that, edge $e_2$ is considered to be inside the community $C_4$, since it is a part of the \textit{rectangle} unit pattern (i.e.,  $\square v_1v_2v_3v_4$), denoted as $c_3$,  which partially intersects with a circle centered at vertex $v_4$ and with radius $r$. $\qquad$  $\blacksquare$
\end{example}

\begin{figure}[t!]
\centering
\subfigure[][{\small A spatial community $C_4$ }]{
\scalebox{0.35}[0.35]{\includegraphics{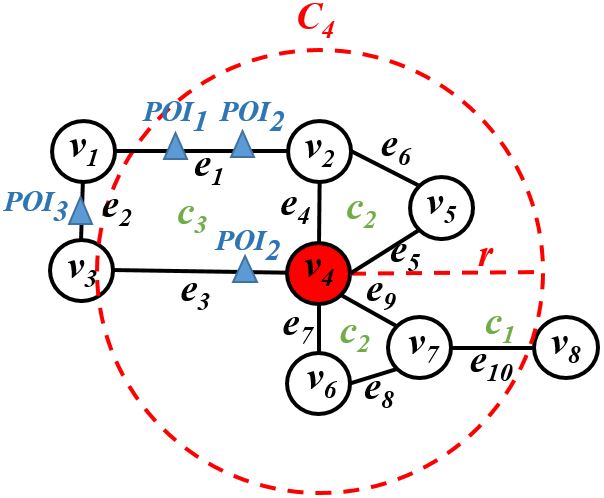}}
\label{subfig:com_example}
}
\subfigure[][{\small POI vector of $C_4$ }]{\hspace{-1ex}
\scalebox{0.3}[0.3]{\includegraphics{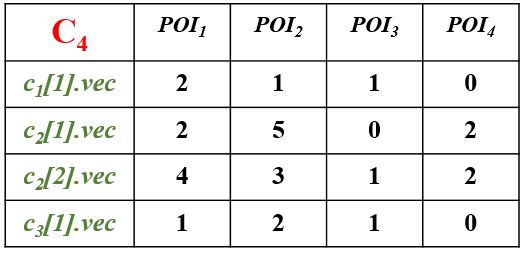}}
\label{subfig:POI_vec}
}
\caption{\small An example of a spatial community.}
\label{fig:community}
\end{figure}

\subsection{Similarity Between Two Communities}

In this subsection, we first propose a similarity metric to measure the similarity between two unit patterns, and then provide the definition of the similarity score between two communities.

\vspace{1ex}\noindent {\bf The Similarity Score Between Unit Patterns.} We first give the definition of the similarity score between two unit patterns. In particular, for two unit patterns of the same type (e.g., delta or rectangle), we define their similarity based on their POIs via the \textit{cosine similarity} \cite{wiki:xxx}.

\begin{definition}
\textbf{(The Similarity Score of Two Unit Patterns)} Assume that we have two unit patterns $c_x$ and $c_y$, whose POI vectors are represented by $c_x.vec$ and $c_y.vec$, respectively. Then, we can compute their similarity score as:
\begin{equation} 
sim(c_x, c_y) = cos\_sim(c_x.vec, c_y.vec), 
\label{eq:structure_score}
\end{equation}
\noindent where function $cos\_sim(c_x.vec, c_y.vec)$ outputs the cosine similarity \cite{wiki:xxx} between vectors $c_x.vec$ and $c_y.vec$.
\end{definition}

In particular, given two vectors $A = (A_1, A_2, \cdots, A_n)$ and $B = (B_1, B_2,$ $\cdots,$ $B_n)$, the cosine similarity, $cos\_sim(A, B)$, in Eq.~(\ref{eq:structure_score}) is given by the normalized dot product of vectors $A$ and $B$ as follows:
\begin{eqnarray} 
cos\_sim (A, B) = \frac{A \cdot B}{||A||\cdot ||B||} = \frac{\sum^n_{h=1} A_h B_h }{ \sqrt{\sum^n_{h=1} A^2_h} \sqrt{\sum^n_{h=1} B^2_h}}.\hspace{-3ex}\notag\\
\label{eq:cosine_sim}
\end{eqnarray}

Note that, in Eq.~(\ref{eq:cosine_sim}), we assume that vectors $A$ and $B$ (or POI vectors of unit patterns in Eq.~(\ref{eq:structure_score})) are normalized to have length 1 (i.e., $||A|| = ||B|| = 1$). As a result, we have:
\begin{eqnarray}
cos\_sim (c_x.vec, c_y.vec) &=& c_x.vec \cdot c_y.vec\\
&=& \sum^n_{h=1} \big(c_x.vec[h] \times c_y.vec[h]\big).\notag
\label{eq:cosine_sim_norm}
\end{eqnarray} 

For example, assume that we have a query unit pattern (edge) $q_1[1]$, whose POI vector is given by $q_1[1].vec$ = $(2, 2, 1 , 1)$. From Figure \ref{fig:community}, we have $c_1[1]$ which is an edge similar to $q_1[1]$, where $c_1[1].vec = (2, 1, 1, 0)$. Now, the similarity score between unit pattern $c_1[1].vec$ and $q_1[1].vec$ can be calculated as $cos\_sim(q_1[1].vec, c_1[1].vec)$, which is equal to 7 ($=2\times 2+2\times 1+1\times 1+1\times 0$).

\vspace{1ex}\noindent {\bf The Similarity Score Between Spatial Communities.} 
The similarity score between a candidate community $C_l$ and a query community, $Q$, can be calculated below.

\begin{definition}
\textbf{(The Similarity Score Between Two Communities).} Given a community $C_l$, a query community, $Q$, and their unit patterns (of the $h$-th type) $c_h \in C_l$ and $q_h \in Q$ $(1 \leq h \leq n)$, the similarity score, $sim(C_l, Q)$, between communities $C_l$ and $Q$ is given by the average cosine similarity of POI vectors of each unit pattern type in $c_h$ and $q_h$, that is,
\begin{eqnarray}
&&sim(C_l, Q) \label{eq:community_score}\\
&=&  \frac{\sum^n_{h=1} sim(c_h,q_h)}{n} \notag \\
         &=& \sum^n_{h=1} \frac{\bigg \{ \frac{\sum_{i=1}^{|c_h|}\sum_{j=1}^{|q_h|} cos\_sim(c_h[i].vec, q_h[j].vec)}{|q_h|}\bigg \}}{n} \notag \\
         &=& \sum^n_{h=1}\frac{\sum_{i=1}^{|c_h|}\sum_{j=1}^{|q_h|} cos\_sim(c_h[i].vec, q_h[j].vec)}{|q_h| \cdot n}, \hspace{-3ex} \notag
\end{eqnarray}
\noindent where $|q_h|$ is the number of unit patterns of the $h$-th shape in the query community $Q$, and $c_h[i]$ (or $q_h[j]$) is the $i$-th (or $j$-th) unit pattern (with the $h$-th shape) in $C_l$ and $Q$, respectively.
\end{definition}

Intuitively, the community $C_l$ may contain $n$ unit pattern types, the $h$-th of which may have $|c_h|$ instances of such a unit pattern type in $C_l$. The case of the query community $Q$ is similar. Thus, in Eq.~(\ref{eq:community_score}), for the $h$-th unit pattern type, we can compute the summed similarity between unit patterns $c_h[i]$ and $q_h[j]$, divide it by $|q_h|$, and then take the average score for all the $n$ unit pattern types.

As an example, in Figures \ref{fig:community} and \ref{fig:communities}, we have a candidate community $C_4$ and a query community $Q$, respectively. Here, $C_4$ and $Q$ have three types of unit patterns, edge $(c_1)$, delta $(c_2)$ and rectangle $(c_3)$, with counts $(1, 2, 1)$ and $(1, 1, 2)$, respectively. Thus, based on Eq.~(\ref{eq:community_score}), the similarity score $sim(C_4, Q)$ can be calculated as the average similarity of all three unit pattern types in $C_4$ and $Q$, that is, $ \frac{sim(c_1, q_1) + sim(c_2, q_2) + sim(c_3, q_3)}{3}$.

\subsection{Top-$k$ Community Similarity Search in Road-Network Graphs}
We next define the problem of \textit{top-$k$ community similarity search} ($Top\text{-}kCS^2$).

\begin{definition}
\textbf{(Top-$k$ Community Similarity Search in Road-Network Graphs, $Top$-$kCS^2$)}  Given a query community $Q$, a road-network graph $G$, a query location $v_q$, and a similarity threshold $\theta$, a \textit{top-$k$ community similarity search} ($Top\text{-}kCS^2$) query retrieves $k$ communities, $C_l$ (for $1 \leq l \leq k$), from $G$, such that:

\begin{itemize}
    \item similarity scores $sim(C_l, Q)$ are greater than or equal to $\theta$ (i.e., $sim(C_l, Q)\geq \theta$, and;
    \item for any community $C_j$ (satisfying $sim(C_j, Q)\geq \theta$ and $C_j\ne C_l$), we have $dist(v_q, C_l)<dist(v_q, C_j)$ (i.e.,     communities $C_l$ are the closest to $v_q$),
\end{itemize}
where the distance function $dist(v_q, C_l)$ is given by the Euclidean Distance \cite{wiki:euclidean} from $v_q$ to the center $v_l$ of the community $C_l$.
\label{def:topkCS2}
\end{definition}

\begin{figure}[t!]
    \centering
    \includegraphics[width=200pt]{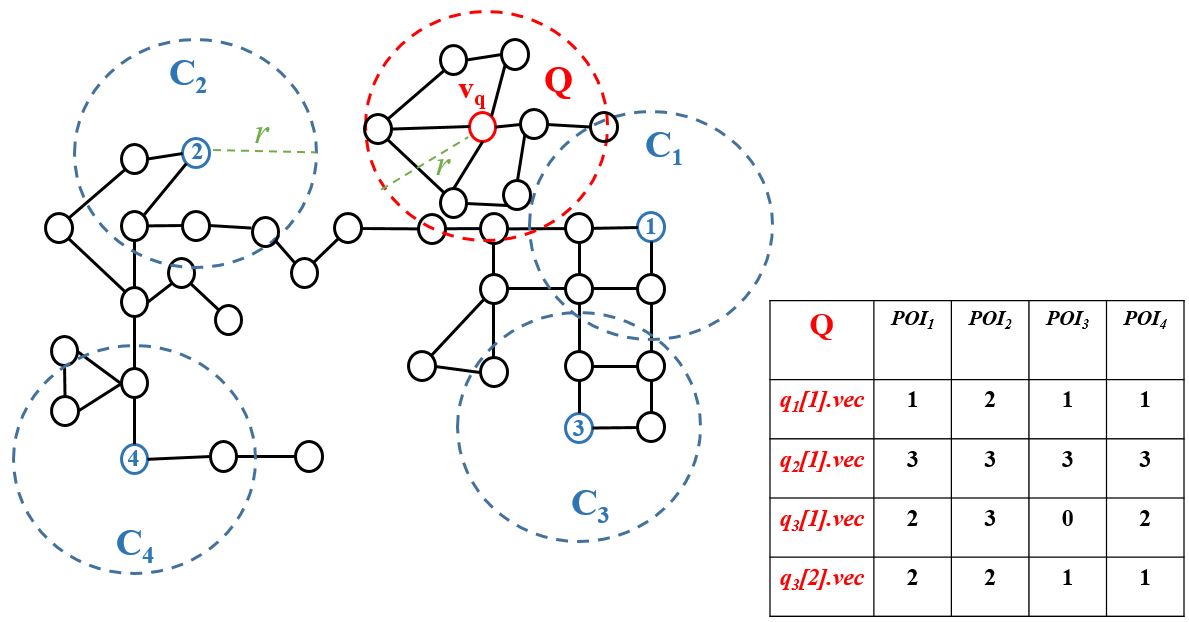}
    \caption{\small An example of communities in a large road-network graph.}
    \label{fig:communities}
\end{figure}

For simplicity, in this paper, we consider the distance, $dist(\cdot, \cdot)$, from the center, $v_q$, of the query community $Q$ to spatial communities $C_l$. Our proposed solution can be easily extended to the one with arbitrary query location $v_q$, not limited to $Q$'s center  (e.g., a new working place in a new city/county). 

As an example in Figure \ref{fig:communities}, we have a query community $Q$, a query vertex, $v_q$, and a radius $r$. In the figure, we have some candidate communities $\{C_1, C_2, C_3, C_4 \}$. Assume that the similarity scores of communities $C_1$, $C_2$, $C_3$, and $C_4$ are 0.7, 0.5, 0.35 and 0.5, respectively. Moreover, the distances (in miles) from $v_q$ to communities $C_1$, $C_2$, $C_3$, and $C_4$ are 0.6, 0.2, 0.55, and 0.4, respectively. If the similarity threshold $\theta$ is $0.5$ and $k =1$, then the $Top\text{-}kCS^2$ problem will return $C_2$ as the answer. This is because the similarity score between $C_2$ and $Q$ is greater than or equal to $0.5$ (i.e., $\theta$) and community $C_2$ has the smallest distance to $v_q$ among communities $C_1 \sim C_4$.

\subsection{Continuous Top-$k$ Community Similarity Search in Road-Network Graphs}

In contrast to the $Top\text{-}kCS^2$ problem with a static query community $Q$, we also define the problem of \textit{continuous top-$k$ community similarity search} ($CTop\text{-}kCS^2$), where the query community is continuously moving between two locations. The $CTop\text{-}kCS^2$ problem can be used in real applications such as the recommendation of communities that are similar to query communities from one's home to working place.

\begin{definition}
\textbf{(Continuous Top-$k$ Community Similarity Search in Large Scale Road-Network Graphs, $CTop\text{-}kCS^2$)}  Given a query line segment $L$ ($= q_{st}, q_{ed}$) and a radius $r$, let a query community, $Q_i$, be a community centered at any point $q_i$ on line segment $L$ with radius $r$. Given a road-network graph $G$, a query location $v_q$, and a similarity threshold $\theta$, a \textit{continuous top-$k$ community similarity search} ($CTop\text{-}kCS^2$) query continuously monitors top-$k$ communities, $C_l$ ($1 \leq l \leq k$), such that:

\begin{itemize}
    \item similarity scores $sim(C_j, Q_i)$ are greater than or equal to $\theta$ (i.e., $sim(C_j, Q_i) \geq \theta$), and;
    \item for any community $C_x$ (satisfying $sim(C_x, Q_i)\geq \theta$ and $C_x \ne C_l$), we have $dist(v_q, C_l)<dist(v_q, C_x)$ (i.e., communities $C_l$ are the closest to $v_q$),
\end{itemize}
where $Q_i$ is the query community with center $q_i \in L$ moving from $q_{st}$ to $q_{ed}$.
\label{def:ctopkCS2}
\end{definition}

Table \ref{tab:notations} depicts the commonly-used notations in this paper and their descriptions.

\begin{table}[htbp]
\caption{Notations and Descriptions}\vspace{-3ex}
\begin{center}
\begin{tabular}{|c||l|}
\hline
\textbf{Notation}&\textbf{Description} \\
\hline\hline
    $o$ & a point of interest (POI)\\\hline
    $c$ & a unit pattern\\\hline
    $C_l$ & a spatial community in road-network graph $G$\\\hline
    $Q$ & a given query community\\\hline
    $n$ & the total no. of unit pattern types\\\hline
    $c_h$ & a unit pattern of type $h$ $(1\leq h \leq n)$ in community $C_l$\\\hline
    $q_h$ & a unit pattern of type $h$ in query community, $Q$\\\hline
    $c_h[i].vec$ & a POI vector for $i$-th unit pattern of $c_h$ ($1 \leq i \leq |c_h|$)\\\hline
    $q_h[j].vec$ & a POI vector for $j$-th unit pattern of $q_h$ ($1 \leq j \leq |q_h|$)\\\hline
    $|c_h|$ & the count of the unit pattern of type $h$ in $C_l$\\\hline
    $|q_h|$ & the count of the unit pattern of type $h$ in $Q$\\
\hline
\end{tabular}
\label{tab:notations}
\end{center}
\end{table}

\section{The Framework for $Top\text{-}kCS^2$ Query Answering}
\label{sec:topk_frame}
In this section, we present a framework for $Top\text{-}kCS^2$ query answering in road-network graphs, $G$, in  Algorithm  \ref{alg:comans_framework}. 

\begin{algorithm}[ht!]
\KwIn{a road-network graph $G$, a similarity threshold  $\theta$, radius $r$, a query community $Q$, and a query vertex $v_q$}
\KwOut{top-$k$ communities, $C_l, (1 \leq l \leq k)$ with similarity scores $\geq \theta$}
\tcp{\hspace{-1ex}Offline Pre-Processing Phase \hspace{-1ex}}
detect all the unit patterns $c \in G$ \qquad \tcp{Algorithm \ref{alg:uc_detect}}
insert all the unit patterns into an aR-tree $I$\\
obtain some communities $C_l (1 \leq l \leq |V|)$ containing each unit pattern and update aggregates in aR-tree $I$ \quad \text{ } \tcp{Algorithm \ref{alg:comm_calc}}
\tcp{\hspace{-1ex}Online Computation Phase \hspace{-1ex}}
\For{each unit pattern type $q_h \in Q$, $1\leq h \leq n$}{
    \For{each unit pattern $q_h[i], 1 \leq i \leq |q_h|$}{
    find a set of unit patterns similar to $q_h[i]$ via index $I$  \tcp{Algorithm \ref{alg:index_traversal}}
    }
}
sort candidate unit patterns based on their similarity scores\\
obtain a list of candidate communities, $cand\_list$, based on the sets of candidate unit patterns w.r.t. $q_h \in Q$\\
\For{each candidate community $C_l \in cand\_list$}{
    calculate an upper bound, $ub\_sim(C_l, Q)$, of the similarity score $sim(C_l, Q)$\\
    \uIf{$ub\_sim(C_l,Q) < \theta$}{
    prune community $C_l$
    }
    \uElse{
    calculate the exact score of $C_l$, $sim(C_l, Q)$\\
    \uIf{$sim(C_l,Q) < \theta$}{
        prune community $C_l$
    }
    \uElse{
     \eIf{$comm\_count < k$}{
            add $C_l$ to a sorted top-$k$ list $ans\_list$\\
            $comm\_count$++
        }{
            \If{$dist(C_l, v_q) < dist(C_k, v_q)$}{
            add $C_l$ to the top-$k$ list $ans\_list$\\
            remove $C_k$ from the top-$k$ list $ans\_list$
            }
        }
    }
    }
}
return the top-$k$ answer list $ans\_list$

\caption{$Top\text{-}kCS^2$ Answering Framework}
\label{alg:comans_framework}
\end{algorithm}

Specifically,  our framework helps retrieve top-$k$ similar communities that satisfy the similarity threshold when compared to the query community, $Q$, and are closer to the query vertex, $v_q$, which consists of \textit{offline pre-processing} and \textit{online computation phases}.

In the \textit{offline pre-processing phase}, we first detect all the unit patterns, $c$, on road networks $G$ (line 1), by invoking our proposed algorithm, $Get\_Unit(G)$, in Algorithm \ref{alg:uc_detect} (as will be discussed in Section \ref{subsec:unit_pattern_detection}). Then, we insert all the unit patterns, $c \in G$, into an aggregate R-tree index, $I$, that is, aR-tree \cite{lazaridis2001progressive}, and offline pre-compute all the communities (with radius $r$) in the road-network graph $G$ (via Algorithm \ref{alg:comm_calc}), whose statistics (e.g., lower/upper bounds of pattern counts) can be used as aggregates for unit patterns in the aR-tree (lines 2-3). 

In the \textit{online computation phase}, for each unit pattern $q_h$ in the query community $Q$, we use the aR-Tree to retrieve a set of similar candidate unit patterns, in descending order of similarity scores (lines 4-7). Next, we use these candidate unit patterns, with respect to $q_h$ ($1 \leq h \leq n$), to obtain a number of candidate communities $C_l$ in a list $cand\_list$ (line 8). For each candidate community, $C_l \in cand\_list$, we first calculate the upper bound similarity score, $ub\_sim(C_l, Q)$ (lines 9-10). If the similarity upper bound score of $C_l$ is less than threshold $\theta$ (i.e. $ub\_sim(C_l, Q) < \theta$), we can safely prune the community $C_l$ (lines 11-12). Otherwise, we calculate the exact similarity score, $sim(C_l, Q)$, for candidate community $C_l$ (line 14). If it holds that $sim(C_l, Q) < \theta$, then we can safely rule out community $C_l$ (lines 15-16). On the other hand, if $sim(C_l, Q)\geq \theta$ holds, we will check whether we have $k$ candidate communities in the current top-$k$ list, $ans\_list$ (lines 17-24). When the count, $comm\_count$, of communities in the current top-$k$ list $ans\_list$ is less than $k$, we can directly add community $C_l$ to this list and increase the count, $comm\_count$, by 1 (lines 18-20). When $comm\_count$ is equal to $k$, we will consider the constraint of the distance of community $C_l$ to query vertex $v_q$. If $C_l$ is closer than the $k$-th closest community $C_k$ in the top-$k$ list $ans\_list$, then we remove community $C_k$ from the list and insert $C_l$ into the top-$k$ list $ans\_list$ (lines 22-24). Finally, we return the top-$k$ answer list, $ans\_list$, after checking all candidate communities in $cand\_list$ (line 25).

\section{Offline Pre-Processing Phase}
\label{sec:offline_processing}

In this section, we discuss the offline pre-processing phase in the framework for our $Top\text{-}kCS^2$ query processing, as given in the first three lines of Algorithm \ref{alg:comans_framework}.

\subsection{Unit Pattern Detection}
\label{subsec:unit_pattern_detection}

In line 1 of the offline computation phase (Algorithm \ref{alg:comans_framework}), we need to detect all the unit patterns, $c$, on road networks $G$. As illustrated in Figure \ref{fig:structures}, in addition to the \textit{edge} unit pattern, each polygon with a simple cycle on the planar graph $G$ is considered as a unit pattern. In order to identify these unit patterns, we present an algorithm, denoted as $Get\_Unit(G)$, which is given by Algorithm \ref{alg:uc_detect}.  

The basic idea of Algorithm $Get\_Unit(G)$ is as follows. For each vertex $v_i\in V(G)$ in road network $G$, we start to traverse the graph from $v_i$ in a clockwise direction (i.e., always choosing the rightmost edges to turn at intersection points) to detect unit patterns containing $v_i$. To avoid traversing the same directed edge multiple times, we utilize a boolean array, $e\_arr[v_x, v_y]$, where each boolean element in the array indicates whether an edge from direction $v_x$ to $v_y$ has been visited before. Similarly, to avoid detecting the same unit patterns multiple times, we use a hashmap, $hash\_unit$, to record unit patterns with their hash values to check whether a unit pattern has been detected before.

\begin{algorithm}[ht]
\KwIn{a road-network graph $G$}
\KwOut{a set of similar unit patterns, $up\_list$, for graph $G$}

$up\_list=\emptyset$\\
initialize the starting vertex, $start$ (= $v_i, 1 \leq i \leq |V|$)\\
initialize an array, $e\_arr$, of the visited states for all directed edges, $v_xv_y$, as not visited (i.e., $e\_arr[v_x, v_y] = false$)\\
initialize hashmap, $hash\_unit$, to store found unit patterns\\
\For{each starting vertex, $start \in V(G)$}{

    load the neighbors, $v_i$, of $start$ to a list, $list$\\
    
    \For{each neighbor $v_i$ in $list$}
    {
        \If{$e\_arr[start, v_i] == true$}{
            continue;   \tcp{the unit pattern has been detected}
        }
        
        $pattern = \{start, v_i\}$\\
        $next = v_i$\\
        $hop\_count = 1$\\

        {\color{black}
        \While{ $next$ $!=$ $NULL$
        }{
            obtain the neighbor vertex $v_j$ of the current vertex $next$ with the smallest angle on a clockwise direction (i.e., $v_j = get\_next(next, V(G), E(G))$)\\
            \If{start == $v_j$}{
            \tcp{the unit pattern has been detected}
            decide the type of the unit pattern $pattern$ based on $hop\_count$\\
            add the unit pattern $pattern$ to $up\_list$ \\    
            add the hash value of unit pattern to the hash bucket, $hash\_unit$\\
            $\forall_{p_i, p_j \in pattern}$ $e\_arr[p_i, p_j] = true$ \\
            $pattern=\emptyset$\\
            break;
            }
            $pattern = pattern \cup\{v_j\}$\\
            $next = v_j$\\
            $hop\_count$++
        }
        }
    }
}
return $up\_list$

\caption{Unit Pattern Detection, $Get\_Unit(G)$}
\label{alg:uc_detect}
\end{algorithm}

\noindent Specifically, in Algorithm \ref{alg:uc_detect}, we first initialize an empty unit pattern list, $up\_list$ (line 1). Then, we use the starting vertex, $start$, to start the graph traversal and detect unit patterns, which can be any vertex $v_i$ from the vertex set $V(G)$ in road-network graph $G$ (line 2). Moreover, we utilize a boolean array $e\_arr$ to set all the directed edges $v_xv_y$ as not visited (i.e., $e\_arr[v_x, v_y] = false$; line 3) and a hashmap, $hash\_unit$ to store the found unit patterns (line 4). Note that, if we always detect unit patterns in a clockwise direction (i.e., always choosing the rightmost turn at intersection points), then each directed edge will only be accessed once. Thus, here, we use each hash value, $get\_hash(pattern)$, in the hashmap, $hash\_unit$ to indicate whether a unit pattern (containing vertices in the array, $pattern$) has been detected before. 

Then, with a for loop (lines 5-23), we iteratively find unit patterns, $pattern$, with different starting vertices $start$ in $G$. In particular,
for each starting vertex $start$, we load all the neighbors, $v_i$, of the starting vertex, $start$, to a list, $list$ (line 6). For each neighbor $v_i \in list$, we check whether or not the unit pattern has been detected before (i.e., $e\_arr[start, v_i] == true$). If the answer is yes, then we will continue to check the next unit pattern (lines 7-9). Otherwise, we add $v_i$ to the array $pattern$ and we assign $v_i$ as $next$ vertex (lines 10-11). In addition, we maintain a hop count, $hop\_count$, for the detected unit pattern $pattern$ (lines 11 and 24). Next, we run a while loop until we hit a $NULL$ (which means there is no path to traverse, line 13) or if we reach the starting node $start$ (line 15). Until there is a path (i.e. $next$ is not $NULL$) we will find the next vertex in a clockwise manner (i.e. $v_j$ = get(next, V(G), E(G)), line 14). If the next vertex is the start vertex, we found a unit pattern, we decide the type of unit pattern based on the $hop_count$ (line 16). We also add the unit pattern to unit pattern list, $up\_list$ (line 17), add the hash value of unit pattern to hash bucket, $hash\_unit$ (line 18), mark all the edges in the array $pattern$ as visited (line 19). Then we empty the array $pattern$ (line 20) and then exit the while loop (line 21). After all unit patterns have been detected in graph $G$, we will return this list, $up\_list$, as the output of the algorithm (line 25).

\subsection{The Community Calculation}

In line 3 of our $Top\text{-}kCS^2$ framework (Algorithm \ref{alg:comans_framework}), we will invoke Algorithm \ref{alg:comm_calc} to obtain all unit patterns in each candidate community $C_i$ (centered at vertex $v_i$ and with radius $r$) via index. In other words, we can retrieve those communities (as well as their statistics) for each unit pattern in graph $G$, and update statistics/aggregates (e.g., lower/upper bounds of pattern counts) in the aR-tree index $I$. The details of aggregates in aR-tree $I$ will be discussed  in Section \ref{subsec:index}.

\begin{algorithm}[t!]
\KwIn{a vertex set, $V(G)$, of graph $G$, radius $r$, and an aR-tree $I$}
\KwOut{a set of communities, $C = \{C_1, C_2, \cdots, C_{|V|}\}$}

\For{each vertex $v_i \in V(G)$ ($1 \leq i \leq |V|$)}{
    load the root of aR-tree $I$ to an empty heap $H$\\
    \While{$H$ is not empty}{
        pop out an MBR node $M$ from the heap $H$\\
        \eIf{$M$ is a leaf node}{
            \For{each unit pattern, $pattern$, in $M$}{
                \If{$pattern$ intersects with a circle centered at $v_i$ with radius $r$}{
                    add $pattern$ to community $C_i$
                }
            }
        }{
            \For{each child node, $M_j$, in $M$}{
                \eIf{$M_j$ is fully inside a circle centered at $v_i$ with radius $r$}{
                    add all unit patterns under node $M_j$ to community $C_i$
                }{
                    \If{$M_j$ partially intersects with a circle centered at $v_i$ with radius $r$}{
                        insert the child node $M_j$ into heap $H$
                    }
                }
            }
        }
    }
} 
return $C = \{C_1, C_2, \cdots, C_{|V|}\}$
\caption{Community Calculation}
\label{alg:comm_calc}
\end{algorithm} 

Specifically, Algorithm \ref{alg:comm_calc} takes a vertex set, $V(G)$, of the graph $G$, a radius $r$, and an aR-tree $I$ of graph $G$. It returns the set of communities $C_i$ for each vertex $v_i\in V(G)$ in graph $G$. In particular, for each vertex $v_i \in V(G)$, we load the root of the aR-tree $I$ into an initially empty heap $H$ (line 2). Next, we traverse the aR-tree to obtain all unit patterns inside community $C_i$, by using heap $H$ (lines 3-15).
If heap $H$ is not empty, we pop out an MBR node from heap $H$ (lines 3-4). If MBR node $M$ is a leaf node, for each unit pattern, $pattern$, in $M$, we check whether or not $pattern$ belongs to the community $C_i$ (i.e., within a circle centered at $v_i$ with radius $r$). If the answer is yes, then we will add $pattern$ to community $C_i$ (lines 5-8). 

On the other hand, if node $M$ is a non-leaf node, we check each entry $M_j$ in node $M$ (lines 9-15). If it holds that $M_j$ is fully inside the circle centered at $v_i$ and with radius $r$, then all unit patterns under node $M_j$ are added to community $C_i$ (lines 11-12). Otherwise, if $M_j$ is partially intersecting with the circle, then we will insert $M_j$ into heap $H$ for further refinement (lines 13-15). Finally, after the index traversals, we obtain and return all candidate communities $C_1 \sim C_{|V|}$ (line 16). 

With these candidate communities, we can calculate the aggregates for those unit patterns in the aR-tree (e.g., lower/upper bounds of pattern counts), which can be used for pruning (as discussed later in Section \ref{sec:prune}).

\section{Pruning Heuristics}
\label{sec:prune}
In this section,  we discuss pruning heuristics for our $Top\text{-}kCS^2$ query. These pruning methods help us prune many false alarms of candidate communities for the $Top\text{-}kCS^2$ query and make the query processing more efficient and effective.

\subsection{Score Upper Bound Pruning} 
To find $k$ similar communities w.r.t a query community, $Q$, in the worst case, we may have to calculate the similarity score, $sim(C_l, Q)$, for each community $C_l \subseteq G$, which is not efficient. To avoid that, here, we propose a pruning technique, which utilizes an upper bound of the similarity score and helps effectively prune many false alarms of candidate communities.

\noindent {\bf Pruning with the Score Upper Bound:} For a given candidate community $C_l$, the similarity score (given by Eq.~(\ref{eq:community_score})) is calculated as the average score of similarities of POI vectors for each unit pattern type. However, it is rather costly and time-consuming to compute such a similarity score for each unit pattern and then take the average score. Therefore, our basic idea of the score upper bound pruning is to obtain an upper bound, $ub\_sim(C_l, Q)$, of the similarity score between communities $C_l$ and $Q$, and filter out those candidate communities with score upper bound less than the threshold $\theta$.

\begin{theorem} {\bf (Score Upper Bound Pruning)} Given a query community $Q$, a community $C_l$, and  a similarity threshold, $\theta$, we can prune a community $C_l$, if $ub\_sim(C_l,Q) < \theta$ holds.
\label{theorem:sim_ub}
\end{theorem}
\begin{proof} Please refer to the proof in Appendix A.
\end{proof}

For example, in Figure \ref{fig:simscore}, we have two vectors of unit pattern type $c_2$, $c_2[1].vec$ and $c_2[2].vec$, and then we calculate the maximum vector, $c_2.max$, which helps us to calculate the upper-bound score. Here, we have a query unit pattern of type $q_2$, $q_2[1].vec$. If we calculate the similarity score, for $c_2[1].vec$ and $q_2[1].vec$, $sim(c_2[1].vec, q_2[1].vec)$ = $2 \times 3 + 5 \times 3 + 0 \times 3 + 2 \times 3 = 27$. Similarly, $sim(c_2[2].vec, q_2[1].vec)$ = $4 \times 3 + 3 \times 3 + 1 \times 3 + 2 \times 3 = 30$. Now, $ub\_sim(c_2, q_2)$ = $sim(c_2[1].max, q_2[1].vec)$ = $4 \times 3 + 5 \times 3 + 1 \times 3 + 2 \times 3 = 36$. Therefore, we know that the max vector will give the maximum/upper-bound similarity score.

\noindent{\bf The Calculation of the Score Upper Bound:} For a given graph, if we start calculating the similarity score for each community, $C_l \subseteq G$, we can get the exact score for each community. However, for a large road-network graph with millions of edges and vertices, there may be a huge number of communities, these calculations are time consuming and expensive. Thus, we propose an upper bound similarity calculation method which will help calculate the upper bound score for the similarity score and then quickly reduce the search space.

\begin{lemma}
Given a community, $C_l$, and a query community, $Q$, the similarity score upper bound, $ub\_sim(C_l, Q)$, is given as:
\begin{equation}
    ub\_sim(C_l, Q) = \sum_{h=1}^{n} cos\_sim(c_h.max, q_h.max),
\label{eq:ub_score}
\end{equation}
where $c_h.max$ and $q_h.max$ are the maximum vector for each unit pattern type $h$, for community $C_l$ and $Q$ respectively.
\label{lemma:ub_sim}
\end{lemma}

Intuitively, in Lemma \ref{lemma:ub_sim}, we compute the maximum vector, $c_h.max$ and $q_h.max$, for each unit pattern in the candidate community $C_l$ and query community $Q$, respectively. By using these two vectors, we can obtain the upper bound of the similarity score.

\begin{figure}[t!]
    \centering
    \includegraphics[width=190pt]{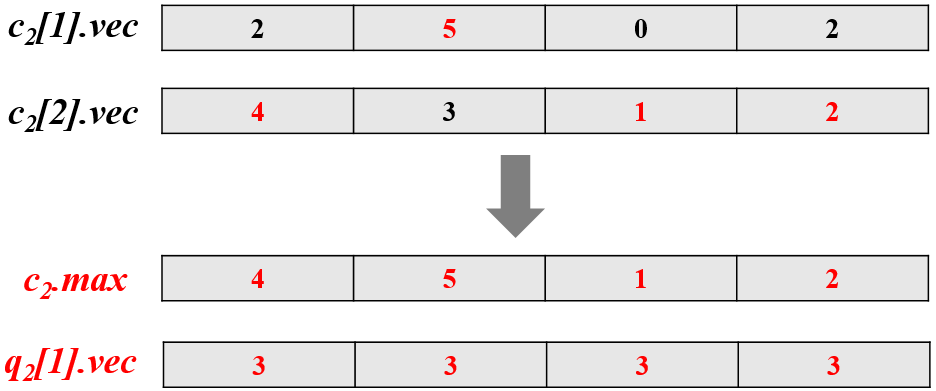}
    \caption{\small Maximum POI vector of unit pattern type $c_2$, $c_2.max$.}
    \label{fig:simscore}
\end{figure}

For example, in Figure \ref{fig:community}, we have a community, $C_4$, with three types of unit pattern $c_1$, $c_2$, and $c_3$. Here, unit pattern type $c_1$ has one unit pattern, $c_1[1]$, similarly, type $c_2$ has two unit patterns, $c_2[1]$ and $c_2[2]$, and $c_3$ has one unit pattern, $c_3[1]$. Each of this unit pattern contains a $POI$ vector represented as $c_1[1].vec$ (for unit pattern $c_1[1]$). Unit pattern type $c_2$ has two $POI$ vectors, $c_2[1].vec$ and $c_2[2].vec$. Thus, the maximum vector for $c_2$, $c_2.max$, is the maximum value vector of $c_2$ as shown in the Figure \ref{fig:simscore}. Similarly, we get the maximum $POI$ vector for each unit pattern type for the query community, $Q$.

\subsection{Distance Pruning}
In the $Top\text{-}kCS^2$ problem, we also consider the distance as an important factor for returning query answers. So, we can also prune communities based on the distance from the query community, $Q$.

\begin{theorem} {\bf (Distance Pruning)}
Assume that $k$ candidate communities $C_1$, $C_2$, ..., and $C_k$ satisfying $sim(C_i, Q) > \theta$ ($1\leq i\leq k$), where $C_k$ has the $k$-th largest distance, $dist(v_q, C_k)$, to the query community $Q$. Then, we can safely prune a community $C_l$, if it holds that $dist(v_q, C_l) \geq dist(v_q, C_k)$.\label{theorem:distancepruning}
\end{theorem}
\begin{proof} Please refer to the proof in Appendix B.
\end{proof}

As an example in Figure \ref{fig:dist_prune}, consider 4 communities $C_1$, $C_2$, $C_3$, and $C_4$, where $k = 2$. Assume that we have obtained two communities $C_1$ and $C_2$, which satisfy the similarity constraint (w.r.t. $\theta$). Then, the second largest distance is $d_2$ ($=dist(v_q, C_2)=0.4$). Based on Theorem \ref{theorem:distancepruning}, we can safely prune communities $C_3$ and $C_4$, since their distances, 0.7 ($=d_3 = dist(v_q, C_3)$) and 0.9 ($=d_4 = dist(v_q, C_4)$), are greater than $d_2$.

\begin{figure}[t!]
\centering
\scalebox{0.4}[0.4]{\includegraphics{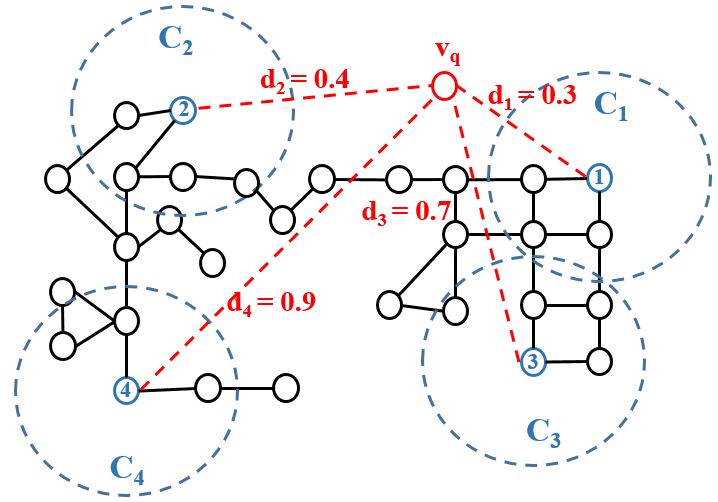}}
\caption{\small Illustration of the distance pruning.}
\label{fig:dist_prune}
\end{figure}

\section{Candidate Unit Pattern Retrieval}
\label{sec:candidate_unit}

In this section, we discuss about the method to obtain the candidate communities $C_l$ for a given query community $Q$, in line 6 of Algorithm \ref{alg:comans_framework}. 

\subsection{Theoretical Analysis}
\label{subsec:theorem}

Let $Q = \{q_1, q_2, \cdots, q_{n} \}$ be the set of unit patterns in $Q$. Here, each unit pattern type $q_h~ (1\leq h \leq n)$ has a set of unit patterns, i.e. $q_h$ = $\{ q_h[1], q_h[2],\cdots, q_h[|q_h|]\}$. The following theorem gives the search radius (w.r.t. query unit pattern $q_h[i]$) to retrieve candidate unit patterns $c_h[i]$.

\begin{theorem}
Given unit patterns $q_h \in Q$ (of the $h$-th type), and a similarity threshold, $\theta$, we say that, a unit pattern $c_h[j]$, is the candidate unit pattern, if the summed cosine similarity, $\sum_{i=1}^{|q_h|}cos\_sim(c_h[j], q_h[i])$ is greater than or equal to $\frac{\theta \cdot |q_h|}{\max\{|c_h|\}}$, that is,

\begin{equation}
 \sum_{i=1}^{|q_h|}cos\_sim(c_h[j], q_h[i]) \geq \frac{\theta \cdot |q_h|}{\max\{|c_h|\}},
 \label{eq:eq6}
\end{equation}
\noindent where $\max\{|c_h|\}$ is the maximum count of unit patterns of type $h$, for all communities containing $c_h[j]$.
\label{theorem:unit_pattern_retrieval}
\end{theorem}
\begin{proof} Please refer to the proof in Appendix C.
\end{proof}

\noindent {\bf Pruning Rule During the Unit Pattern Retrieval:} Intuitively, from Theorem \ref{theorem:unit_pattern_retrieval}, we can obtain a similarity threshold (i.e., $\frac{\theta \cdot \ |q_h|}{\max\{|c_h|\}}$) for each type of query unit patterns $q_h$ in the query community $Q$. Given a unit pattern $c_h[j]$, we can compute an upper bound of the summed similarity score $\sum_{i=1}^{|q_h|}cos\_sim(c_h[j], q_h[i])$, and rule out the unit pattern $c_h[j]$ if the score upper bound is less than the threshold $\frac{\theta \cdot |q_h|}{ \max\{|c_h|\}}$, where $\max\{|c_h|\}$ can be offline calculated during the pre-computation phase.

\subsection{Unit Pattern Retrieval Algorithm via Index Traversal}
\label{subsec:index}

\noindent{\bf The aR-tree Index $I$:} As mentioned in lines 2-3 of Algorithm \ref{alg:comans_framework} (framework), we build an aR-tree index \cite{lazaridis2001progressive}, $I$, over the road-network graph $G$, which can facilitate the $Top\text{-}kCS^2$ query answering. Specifically, for each unit pattern, $c_h \in G$, we use a \textit{minimum bounding rectangle} (MBR) to minimally bound $c_h$ (as shown in Figure \ref{subfig:MBR}), and then insert the MBR into the aR-tree. As a result, the aR-tree structure is as follows.

\underline{\it Leaf Nodes:} Each leaf-node in the aR-tree $I$ contains unit patterns. These unit patterns are the smallest units extracted from the road-network graph $G$ (as discussed in Algorithm \ref{alg:uc_detect}). 

As an example in Figure \ref{subfig:Rtree}, brown nodes at the bottom of the aR-tree correspond to leaf nodes.

\underline{\it Non-Leaf Nodes:} Each non-leaf node in the aR-tree $I$ contains several entries of child nodes. Each entry in the non-leaf node stores an MBR which minimally bounds all MBRs (or unit patterns) under this entry, as well as the summary information, $INFO =$ $(MIN, MAX, ARR, VMAX)$ of all children under entry. Here, $MIN$ and $MAX$ refer to the minimum and maximum corner points of the MBR on $x$- and $y$-axes, respectively; $ARR$ is a bit vector of POIs which indicates whether or not an entry contains some POIs; $VMAX$ is a $n \times m$ matrix, which stores the maximum count vector for each unit pattern type, $c_h$ ($1 \leq h \leq n$).

Figure \ref{fig:index} illustrates an example of aR-tree with non-leaf nodes, where MBRs $E_1 \sim E_3$ (green nodes) correspond to non-leaf nodes. For example, in Figure \ref{subfig:Rtree}, MBR $E_3$ stores the summary information $E_3.INFO$ = $(MIN, MAX, ARR, VMAX)$. Here, $MIN$ $= (MIN_X, MIN_Y)$, where $MIN_X$ = $\min \{c_1.MIN_X, c_2.MIN_X, c_3.MIN_X\}$ and $MIN_Y$ = $\min \{c_1.MIN_Y, c_2.MIN_Y, c_3.MIN_Y\}$. Similarly, we can obtain $MAX = (MAX_X, MAX_Y)$, where $MAX_X$ = $\max \{c_1.MAX_X, c_2.MAX_X, c_3.MAX_X\}$ and $MAX_Y$ = $\max \{c_1.MAX_Y, c_2.MAX_Y, c_3.MAX_Y\}$. $ARR$ is a bit vector indicating whether or not a $POI$ is contained in that MBR, e.g., $ARR$ = $\{1,1,1,1\}$. $VMAX$ is the matrix which indicates the maximum POI vector for each structure type, e.g., in Figure \ref{fig:simscore}, $VMAX_2$ = $\{4, 5, 1, 2\}$.

\noindent {\bf Index Traversal for Retrieving Similar Unit Patterns:} In line 6 of our $Top\text{-}kCS^2$ framework (Algorithm \ref{alg:comans_framework}), we need to obtain candidate unit patterns that are similar to a query unit pattern $q_h[i]$, by traversing the aR-tree $I$ via Algorithm \ref{alg:index_traversal}. 

Specifically, Algorithm \ref{alg:index_traversal} retrieves a set of candidate unit patterns for each type of query unit pattern set $q_h$ (line 1). First, we insert the root of the aR-tree, $I$, into heap $H$ (line 2). Next, we traverse the index $I$ to find all the candidate unit patterns (line 3-12). Each time we de-heap a node $M$ from heap $H$ (lines 3-4). When node $M$ is a leaf node, for each unit pattern $c_h[j]$ in $M$, we will check whether $\sum_{i=1}^{|q_h|} cos\_sim(c_h[j], q_h[i]) \geq t$ holds (from Theorem \ref{theorem:unit_pattern_retrieval}). If the answer is yes, then we will add the unit pattern $c_h[j]$ to the candidate set $c_h$ (lines 5-8). When node $M$ is a non-leaf node, for each child $M_l$ in $M$, we compute an upper bound, $\sum_{i=1}^{|q_h|} ub\_sim(M_l.max, q_h[i])$, of the summed similarity score. If $\sum_{i=1}^{|q_h|} ub\_sim(M_l.max, q_h[i]) \geq t$ holds, then we will insert node $M_l$ into heap $H$ for further checking (lines 9-12). The while loop terminates, when heap $H$ becomes empty (line 3). Finally, we return a set, $c_h$, of candidate unit patterns that are similar to $q_h$.

\begin{figure}[t!]
\centering\vspace{-2ex} 
\subfigure[][{\small MBRs}]{
\scalebox{0.3}[0.3]{\includegraphics{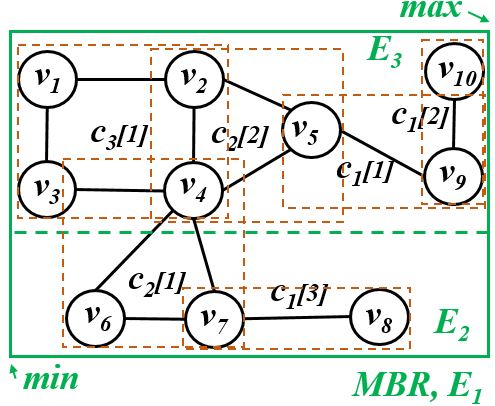}
}
\label{subfig:MBR}
}
\subfigure[][{\small tree structure}]{\hspace{-1ex}
\scalebox{0.3}[0.3]{\includegraphics{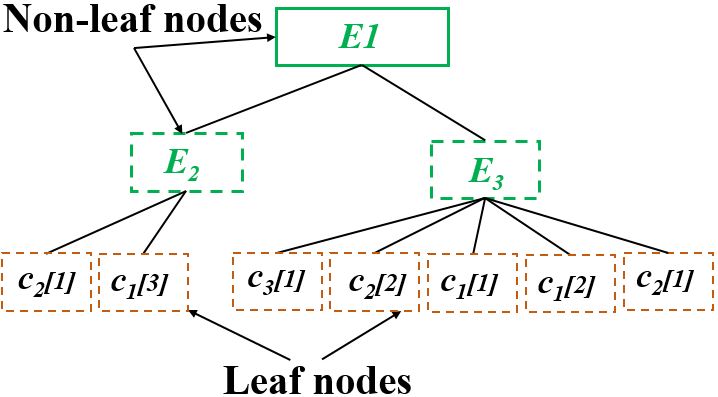}}
\label{subfig:Rtree}
}
\caption{\small An example of the aR-tree index.}
\label{fig:index}
\end{figure}

\begin{algorithm}[ht!]
\KwIn{an aggregate R-tree $I$, query unit pattern sets $q_h$, and a similarity threshold $t$ = $\frac{\theta \cdot \ |q_h|}{\max\{|c_h|\}}$}
\KwOut{a set of candidate unit patterns, $c_h$, similar to unit pattern set $q_h$}
\For{each type of query unit pattern set $q_h$}{
    insert the root of index $I$ into heap $H$\\
    \While{$H$ is not empty}{
        pop out a node $M$ from heap $H$\\
        \eIf{$M$ is leaf node}{
            \For{each unit pattern $c_h[j]$ in $M$}{
                \If{$\sum_{i=1}^{|q_h|} cos\_sim(c_h[j], q_h[i]) \geq t$}{
                    add $c_h[j]$ to the candidate set $c_h$ \tcp{Theorem \ref{theorem:unit_pattern_retrieval}}
                }
            }
        }{
            \For{each child node, $M_l$, in $M$}{
                \If{$\sum_{i=1}^{|q_h|} ub\_sim(M_l.max, q_h[i]) \geq t$}{
                    insert child node $M_l$ into heap $H$ \tcp{Theorem \ref{theorem:unit_pattern_retrieval}}
                }
            }
        }
    }
}
return $c_h$ = $\{c_1, c_2, \cdots, c_z\}$
\caption{Unit Pattern Retrieval via Index Traversal}
\label{alg:index_traversal}
\end{algorithm}

\section{Continuous $Top\text{-}kCS^2$ Processing}
\label{sec:ctopksol}

In this section, we discuss how to tackle the continuous $Top\text{-}kCS^2$ problem (i.e., $CTop\text{-}kCS^2$ as defined in Definition \ref{def:ctopkCS2}), where the query community $Q_i$ continuously moves along line segment $L$.

\subsection{Finding Splitting Points}

Different from $Top\text{-}kCS^2$, in the $CTop\text{-}kCS^2$ problem, the center point, $q_i$, of the query community $Q_i$ can be located anywhere on line segment $L$. It is not efficient to enumerate all possible query communities $Q_i$ with respect to an infinite number of center points $q_i$ on $L$. Thus, we propose to split the query line segment, $L$, into multiple intervals, such that query unit patterns for any center position $q_i$ within each interval remain the same.

\begin{figure}[ht!]
\centering
\scalebox{0.45}[0.45]{\includegraphics{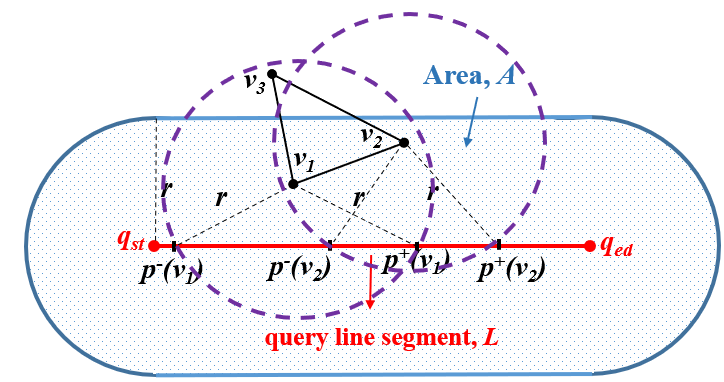}}
\caption{\small Illustration of split points in the query line segment, $L =q_{st}q_{ed}$.}
\label{fig:cont_comm}
\end{figure}

Below, we discuss how to find a list, $S$, of $n$ splitting points, $s_1$, $s_2$, $...$, and $s_n$, on query line segment, $L = q_{st} q_{ed}$, such that when the center point $q_i$ passes by each splitting point, the query community $Q_i$ (with a set of unit patterns) will change.

\begin{algorithm}[ht!]
\KwIn{a query line segment $L = q_{st} q_{ed}$, a graph $G$, and a radius $r$}
\KwOut{a list, $S$, of splitting points on $L$}
find all unit patterns intersecting with a circle of radius $r$ whose center moves from $q_{st}$ to $q_{ed}$ on $L$\\
add the intersecting unit patterns to a list, $unit\_list$\\
\For{each unit pattern, $c_h$, in $unit\_list$}{
    $interval\_list = \emptyset$\\
    \For{each vertex $v_i \in c_h$}{
        draw a circle $\odot v_i$ centered at $v_i$ and with radius $r$\\
        obtain two intersection points $[p^-(v_i), p^+(v_i)]$ between $\odot v_i$ and $L$\\ \tcp{if intersection points are out of boundary points of $L$, replace them with $q_{st}$ or $q_{ed}$}
        add interval $[p^-(v_i), p^+(v_i)]$ to a list $interval\_list$
    }
    union all intervals in $interval\_list$ and add boundary points of the resulting intervals (associated with $c_h$) to $S$
}
    
return $S = \{s_1, s_2, \cdots, s_n\}$
\caption{Finding splitting points on a query line segment, $L$}
\label{alg:finding_split}
\end{algorithm}

Algorithm \ref{alg:finding_split} presents the process of finding splitting points on the query line segment, $L = q_{st}q_{ed}$. We first obtain all unit patterns that intersect with a circle with radius $r$, whose center moves from $q_{st}$ to $q_{ed}$ along $L$ (line 1). Next, we add all these intersecting unit patterns to a list, $unit\_list$  (line 2). For each unit pattern $c_h$ in $unit\_list$, we draw a circle $\odot v_i$ centered at $v_i$ and with radius $r$ for each vertex $v_i\in c_h$ (lines 3-6). Then, we obtain intervals $[p^-(v_i), p^+(v_i)]$ whose bounds are intersection points between circles $\odot v_i$ and line segment $L$ (line 7). We add intervals $[p^-(v_i), p^+(v_i)]$ to a list $interval\_list$ (line 8). Next, we union all intervals in $interval\_list$ and add boundary points of union intervals to $S$ (line 9). Finally, we return all splitting points $\{s_1, s_2, \cdots, s_n\}$ in $S$.

As an example in Figure \ref{fig:cont_comm}, a query community $Q_i$ (i.e., a circle centered on $L$ and with radius $r$) moves along a query line segment $L = q_{st}q_{ed}$, which forms a shaded area $A$. Consider a unit pattern of triangular shape $\Delta v_1 v_2 v_3$. We draw circles centered at $v_1$ and $v_2$ with radius $r$, and obtain intersection points on $L$, that is, $[p^-(v_1), p^+(v_1)]$ and $[p^-(v_2), p^+(v_2)]$. As shown in the figure, we union these two intervals and obtain an interval $[p^-(v_1), p^+(v_2)]$. Here, $p^-(v_1)$ and $p^+(v_2)$ are splitting points on $L$, and we add them to $S$.

\subsection{Finding Query Communities}
\label{subsec:find_query_communities}

Intuitively, the splitting points on $L$ indicate the changes of unit patterns in query communities $Q_i$. After obtaining all splitting points in $S$ (as shown in Algorithm \ref{alg:finding_split}), we can incrementally maintain a set of unit patterns for $Q_i$. Specifically, we first obtain a set, $U_0$, of unit patterns for the center of $Q_1$ in interval $[q_{st}, s_1]$. Then, for each next splitting point $s_i$ ($1\leq i\leq n$), we maintain the set, $U_i$, of unit patterns for interval $[s_i, s_{i+1}]$. That is, if the splitting point $s_i$ is a lower bound of $[p^-(v_i), p^+(v_i)]$ (see line 8 of Algorithm \ref{alg:finding_split}), we obtain $U_i = U_{i-1} \cup \{c_h\}$; similarly, if $s_i$ is an upper bound of $[p^-(v_i), p^+(v_i)]$, we have $U_i = U_{i-1} - \{c_h\}$, where $v_i$ is a vertex in unit pattern $c_h$. This way, we can incrementally obtain a set $U_i$ of unit patterns for each query community $Q_i$ with interval $[s_i, s_{i+1}]$ on $L$.

\subsection{$CTop\text{-}kCS^2$ Query Processing}

In this subsection, we present the algorithm for $CTop\text{-}kCS^2$ query answering in Algorithm \ref{alg:conti_topk}. Specifically, we first find a set, $S$, of splitting points on $L$ (as mentioned in Algorithm \ref{alg:finding_split}) (line 1), which will result in intervals $[s_i, s_{i+1}]$. Then, as discussed in Section \ref{subsec:find_query_communities}, we incrementally obtain a set, $U_i$, of unit patterns for each interval $[s_i, s_{i+1}]$ (line 2). We will then apply our proposed $Top\text{-}kCS^2$ algorithm (lines 4-8 in Algorithm \ref{alg:comans_framework}), which traverses the index to retrieve candidate unit patterns for each query unit pattern in $\cup_{\forall i} U_i$ in a batch manner (line 3). Finally, we assemble candidate unit patterns for each query community $Q_i$ (w.r.t., interval $[s_i, s_{i+1}]$ on $L$), and obtain/return top-$k$ communities in $list\_topk_i$ (as given in lines 9-25 of Algorithm \ref{alg:comans_framework}) (lines 4-6).

\begin{algorithm}[ht!]
\KwIn{a query line segment $L = q_{st} q_{ed}$, a graph $G$, and a radius $r$}
\KwOut{top-$k$ communities for each query community $Q_i$ on $L$}
find a set, $S$, of splitting points on $L$ \qquad \tcp{Algorithm \ref{alg:finding_split}}
obtain a set, $U_i$, of unit patterns in query communities $Q_i$ with centers in intervals w.r.t. splitting points in $S$ \tcp{Section \ref{subsec:find_query_communities}}
find candidate unit patterns for each query unit pattern in $\cup_{\forall i} U_i$ \qquad \tcp{lines 4-8 in Algorithm \ref{alg:comans_framework}}
\For{each query community $Q_i$ w.r.t. an interval on $L$}{
    find top-$k$ similar communities, $list\_topk_i$ \qquad \tcp{lines 9-25 in Algorithm \ref{alg:comans_framework}}
}
return $list\_topk_i$
\caption{$CTop\text{-}kCS^2$ query answering}
\label{alg:conti_topk}
\end{algorithm}

\section{Experimental Evaluation}
\label{sec:experiment}

In this section, we verify the effectiveness and efficiency of our proposed $Top\text{-}kCS^2$ and $CTop\text{-}kCS^2$ algorithms over both real and synthetic road-network graphs. 

\subsection{Experimental Settings}

\noindent \textbf{Real/synthetic data sets.} We used both real and synthetic data sets for our experimental evaluation. Specifically, for real  data set, we use the California Road Network \cite{ca_data}, denoted as $CA$, which contains 21,048 road intersection points, 21,693 road segments, and 104,770 \textit{points of interests} (POIs). $CA$ is originally obtained from Digital Chart of the World Server and U.S. Geological Survey. Each vertex in $CA$ data set is represented by (longitude, latitude).

For synthetic data, we first generate vertices of a road-network graph on a spatial data space, following either the Uniform or Clustered distribution. For the Uniform distribution, we generate vertices uniformly in a designated spatial data space; for the clustered data set, we first randomly obtain seed vertices in a spatial space, and then generate other vertices close to these seeds. Here, the clustered data set can simulate dense road networks (i.e., clusters of vertices) in cities. Next, we connect vertices via edges (road segments) on road networks, that is, linking each vertex to $d \in [deg_{min}, deg_{max}]$ random nearest neighbors nearby (avoiding road intersections on the planar graph). This way, we can obtain a random road-network graph, $G$, with an average degree $deg$. By using different spatial distributions of vertices, we produce two types of graph, $uniform$ and $cluster$. 

\noindent \textbf{Measures.} To evaluate the query performance, we selected 15 random center points $v_q$ from road networks to generate query communities for $Top\text{-}kCS^2$ (or random line segments of length $|L|$ for $CTop\text{-}kCS^2$). We report the \textit{pruning power}, \textit{wall clock time}, and \textit{I/O cost}. Here, the \textit{pruning power} is the percentage of candidate communities that are pruned by our pruning strategies; the \textit{wall clock time} is the average time cost to answer $Top\text{-}kCS^2$ queries; the \textit{I/O cost} is the number of node accesses in the aR-tree. 

\noindent \textbf{Competitor.} To our best knowledge, no prior works studied the top-$k$ community search problem in large-scale road-network graphs, which has different community semantics from that on social networks. Thus, in this paper, we compare our $Top\text{-}kCS^2$ approach with a baseline algorithm, named $baseline$, which is a naive approach without using any index. In particular, the $baseline$ method first scans the road-network graph $G$ to retrieve unit patterns from $G$ that are similar to query unit patterns in the query community $Q$, and then computes $k$ communities (containing the retrieved unit patterns) that satisfy the similarity threshold, $\theta$ and are closest to $Q$.

For our $CTop\text{-}kCS^2$ query, we use a naive baseline algorithm, \textit{baseline}, which first retrieves all query communities on $L$ and then obtains top-$k$ communities for each query community without using our proposed index and pruning strategies.

\noindent \textbf{Parameter settings.} Table \ref{table:parameter_conti} depicts the parameter settings, where default values are in bold. Each time we vary the values of one parameter, while other parameters are set to their default values. We ran all the experiments on a machine with Intel Core i7-6600U 2.60GHz CPU, Windows 10 OS, and 512 GB memory. All algorithms were implemented in C++.

\begin{table}[ht!]
\scriptsize
\caption{\small The parameter settings.}
\begin{center}
 \begin{tabu}{| c | l |}
 \hline
 {\bf Parameters} & \qquad\qquad {\bf Values}  \\
 \hline\hline
  $k$ & 1, 5, \textbf{10}, 15, 20 \\
 \hline
  $deg$ & 2, \textbf{3}, 4\\
\hline 
  $r$  & 0.1, 0.5, \textbf{1}, 1.5, 2\\
 \hline 
  $\theta$ &  0.5, 0.55, \textbf{0.6}, 0.65, 0.7\\
 \hline
  $|V(G)|$ & 10K, 20K, \textbf{30K}, 50K, 100K\\
\hline
  $|L|$ & 2, \textbf{4}, 6\\
 \hline
\end{tabu}%\vspace{-2ex}
\label{table:parameter_conti}
\end{center}
\end{table}

\begin{figure}[ht!] 
\subfigure[][{\small pruning power vs. $k$}]{
\scalebox{0.2}[0.2]{\includegraphics{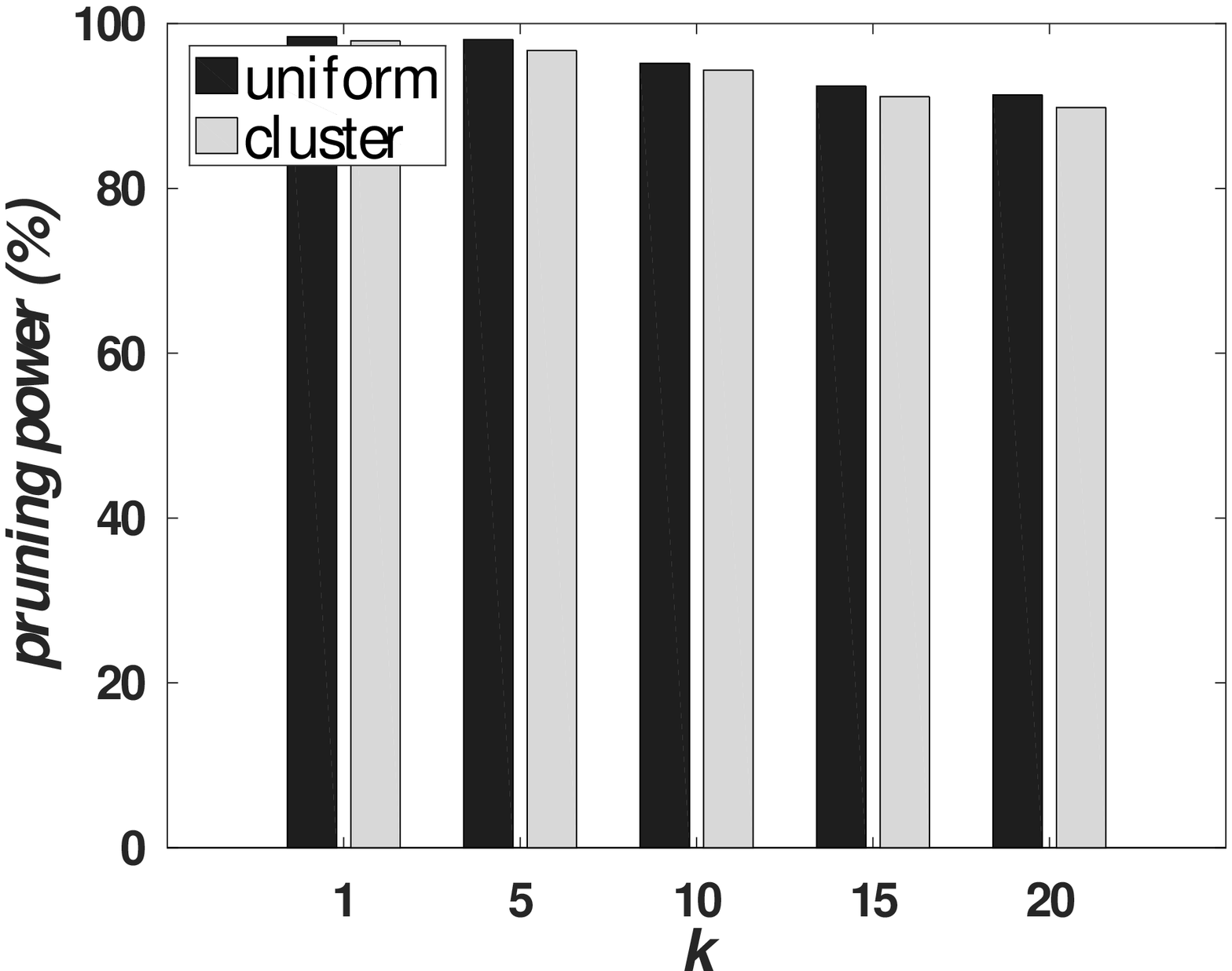}}\label{subfig:k_pp}
}
\subfigure[][{\small pruning power vs. $|V(G)|$}]{
\scalebox{0.2}[0.2]{\includegraphics{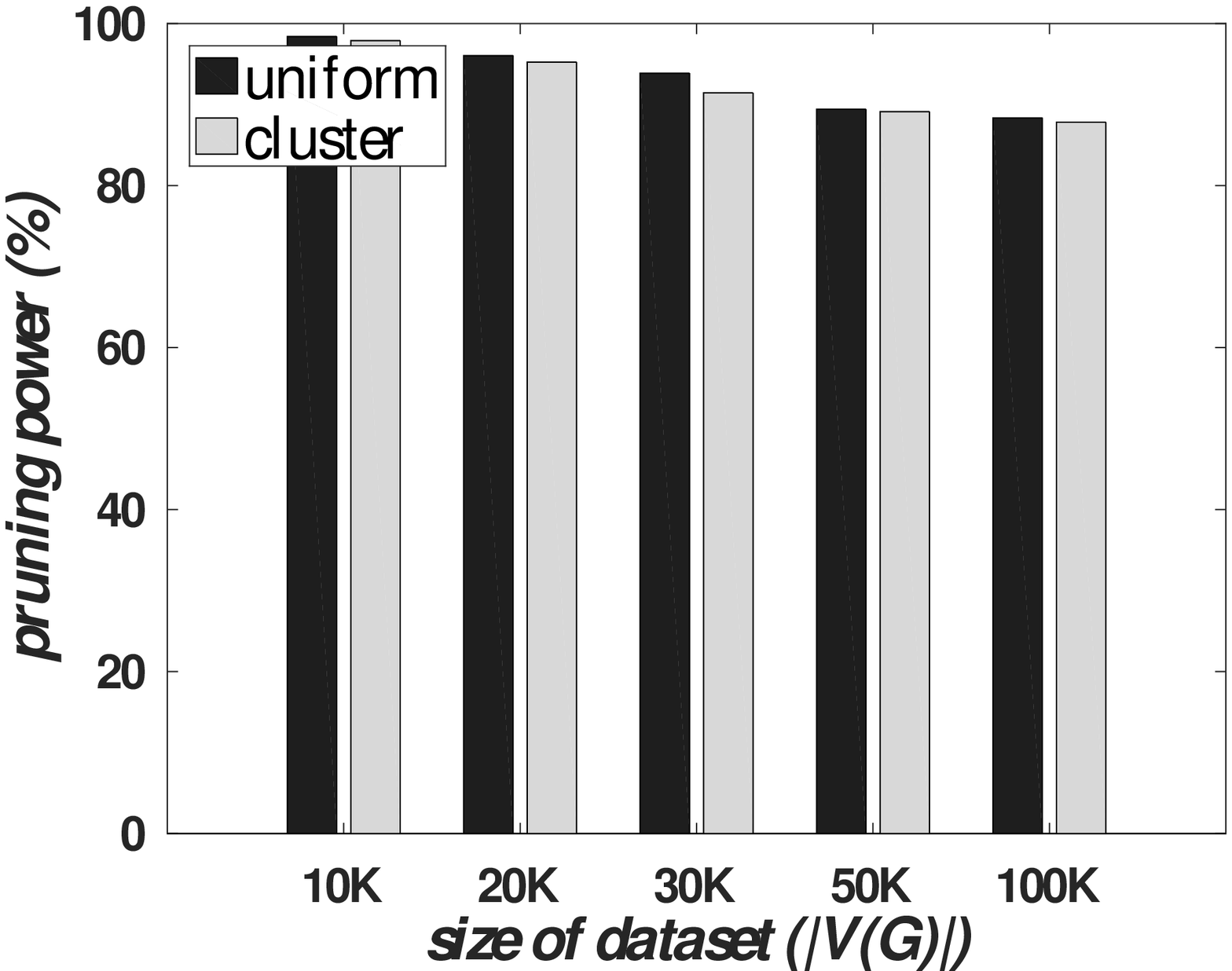}}\label{subfig:data_pp}
}\\
\subfigure[][{\small pruning power vs. $deg$}]{
\scalebox{0.2}[0.2]{\includegraphics{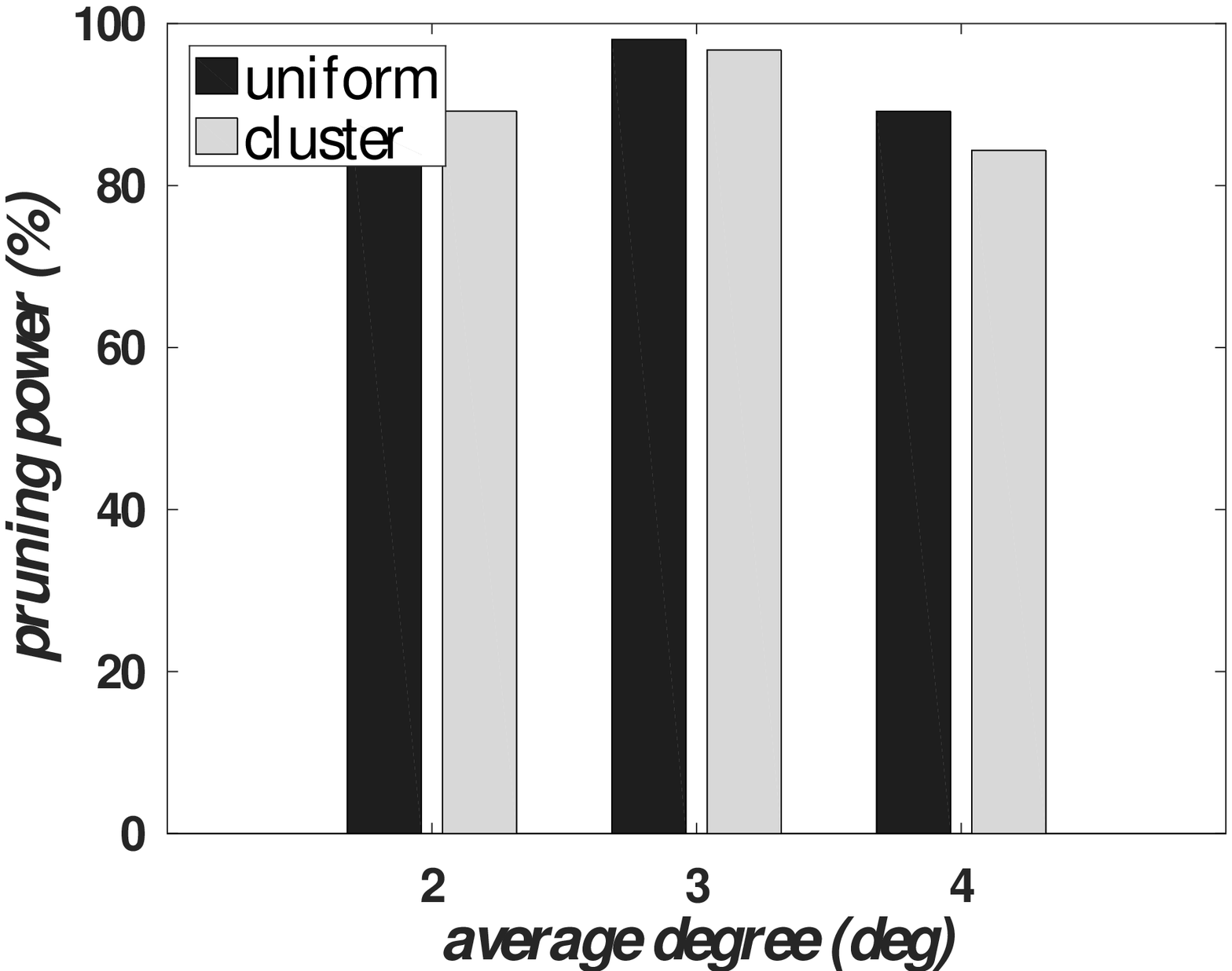}}\label{subfig:deg_pp}
}
\subfigure[][{\small pruning power vs. $r$}]{
\scalebox{0.2}[0.2]{\includegraphics{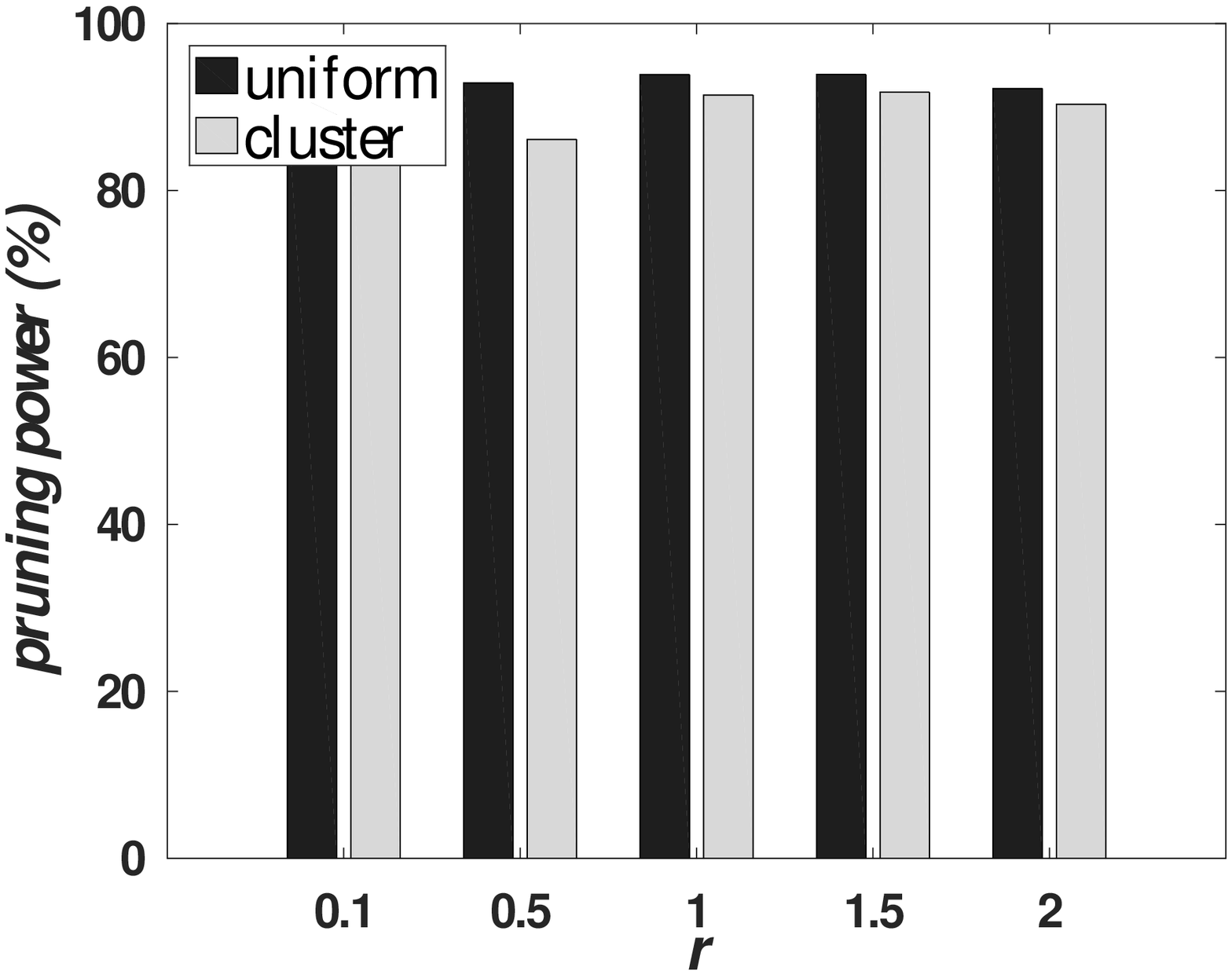}}\label{subfig:rad_pp}
}\\
\centering
\subfigure[][{\small pruning power vs. $\theta$}]{
\scalebox{0.2}[0.2]{\includegraphics{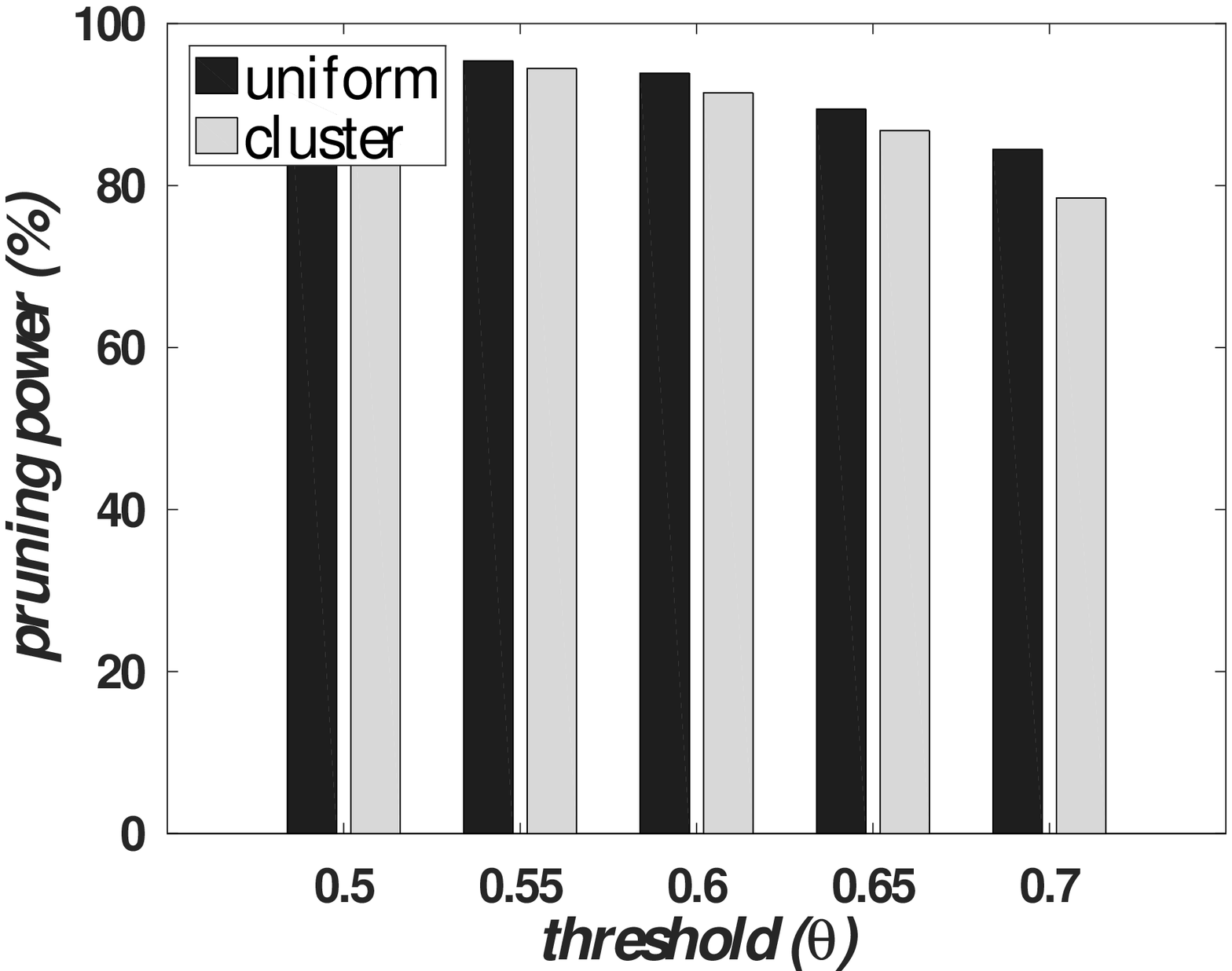}}\label{subfig:theta_pp}
}\\
\caption{\small The $Top\text{-}kCS^2$ pruning power vs. parameters.}
\label{fig:pruning_power}
\end{figure}

\subsection{Evaluation of the $Top\text{-}kCS^2$ Pruning Power}

First, we report the effectiveness of our proposed pruning strategies (i.e., score upper bound and distance pruning as discussed in Section \ref{sec:prune}) over $uniform$ and $cluster$ data sets in Figure \ref{fig:pruning_power}, in terms of the \textit{pruning power}. Specifically, Figure \ref{subfig:k_pp} varies parameter $k$ from 1 to 20 (other parameters are set to default values). For larger $k$, the pruning power slightly decreases, but remains high (i.e., $85\% \sim 94\%$). Figure \ref{subfig:deg_pp} shows the pruning power of our pruning methods with different average degrees $deg$ from 2 to 4, where the pruning power remains high (i.e., $83\% \sim 89\%$) for different $deg$ values. Figure \ref{subfig:rad_pp} reports the pruning power of our pruning methods, by varying the radius $r$ from 0.1 to 2. From the figure, the pruning power is not very sensitive to radius $r$, which is high (i.e., $85\% \sim 93\%$). Figure \ref{subfig:theta_pp} evaluates the pruning power by varying similarity threshold, $\theta$, from 0.5 to 0.7, which gradually decreases for larger $\theta$. Nevertheless, for different $\theta$ values, the pruning power remains high (i.e., $78\% \sim 95\%$). Figure \ref{subfig:data_pp} tests the scalability of our pruning methods for different graph sizes, $|V(G)|$, from $10K$ to $100K$. The pruning power slightly decreases for larger graph, but remains high (i.e., $82\% \sim 94\%$), which confirms the effectiveness of our pruning methods for $Top\text{-}kCS^2$ queries.

\begin{figure}[t!]
\centering 
\scalebox{0.16}[0.16]{\includegraphics{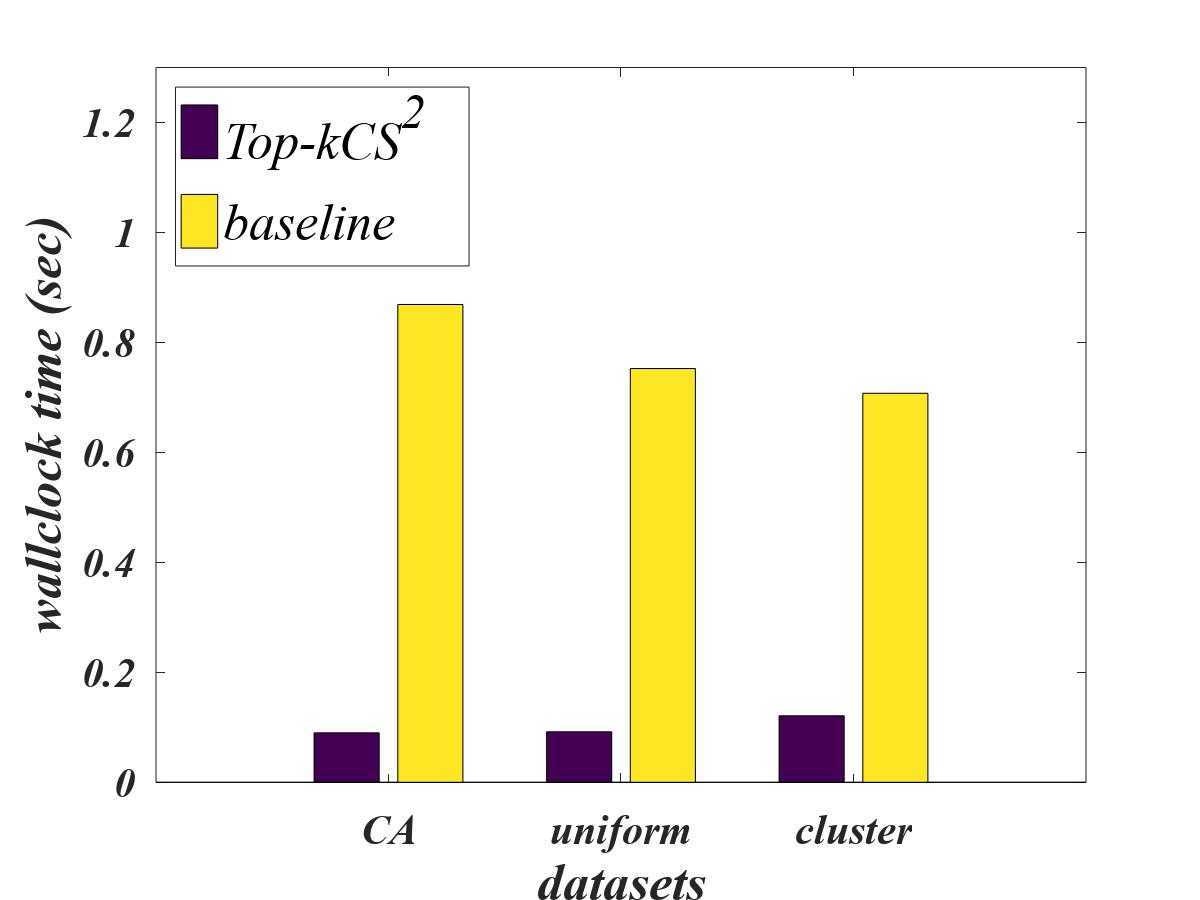}}
\caption{\small The $Top\text{-}kCS^2$ performance vs. real/synthetic data sets.}
\label{fig:datasets}
\end{figure}

\subsection{Evaluation of the $Top\text{-}kCS^2$ Query Efficiency}

\noindent {\bf The $Top\text{-}kCS^2$ performance vs. real/synthetic data sets.} Figure \ref{fig:datasets} compares our $Top\text{-}kCS^2$ approach with the $baseline$ algorithm over real/synthetic data sets, in terms of the wall clock time. From the figure, we can see that the efficiency of the $Top\text{-}kCS^2$ query outperforms that of $baseline$ for all the three data sets. This is because $Top\text{-}kCS^2$ applies effective pruning methods with the help of the index. The experimental results confirm the effectiveness of our pruning methods, and efficiency of our $Top\text{-}kCS^2$ approach. 

Next, we will test the robustness of our $Top\text{-}kCS^2$ approach over synthetic data sets, by varying different parameters (e.g., $k$, $deg$, $r$, $\theta$, and $|V(G)|$). 

\noindent \textbf{The $Top\text{-}kCS^2$ performance vs. parameter $k$.} Figure \ref{fig:k} illustrates the effect of parameter $k$ on the $Top\text{-}kCS^2$ query performance, where $k = 1, 5, 10, 15, 20$, and other parameters are set to default values. When we increase $k$, the wall clock time smoothly increases, whereas the I/O cost also slightly increases. Nonetheless, both wall clock time and I/O cost remain low (i.e., $0.13 \sim 0.3$ $sec$ and $201 \sim 308$ I/Os, respectively), which indicates the query efficiency of our proposed $Top\text{-}kCS^2$ approach with different $k$ values.

\noindent \textbf{The $Top\text{-}kCS^2$ performance vs. average degree $deg$.} Figure \ref{fig:avg_deg} varies the average degree of each vertex in the graph $G$ from 2 to 4, where other parameters are set to default values. When the average degree, $deg$, is small, there are many similar unit patterns (with simple structures), which leads to high processing cost. On the other hand, when the average degree, $deg$, becomes larger, more complex unit patterns (e.g., rectangles and triangles) need to be processed, which requires higher cost. Therefore, for larger average degree, $deg$, as shown in Figure \ref{subfig:time_degree}, the wall clock time first decreases, and then increases. Regarding the I/O cost, since larger $deg$ indicates that we need to access more vertices/edges to search for candidate patterns, the I/O cost increases for larger $deg$ (as shown in Figure \ref{subfig:IO_degree}). Nevertheless, the overall wall clock time and I/O cost remain low (i.e., less than $0.3$ $sec$ and less than 550 I/Os, respectively).

\noindent \textbf{The $Top\text{-}kCS^2$ performance vs. radius $r$.} Figure \ref{fig:radius} shows the wall clock time and I/O cost of our $Top\text{-}kCS^2$ approach over $uniform$ and $cluster$ data sets, by varying the radius $r$ from 0.1 to 2, where default values are used for other parameters (as depicted in Table \ref{table:parameter_conti}). The experimental results show that the wall clock time and I/O cost are not very sensitive to radius $r$. Moreover, for different radius $r$, the wall clock time and I/O cost remain low (i.e., $0.15 \sim 0.25$ $sec$ and $228 \sim 301$ I/Os, respectively).

\noindent \textbf{The $Top\text{-}kCS^2$ performance vs. similarity threshold, $\theta$.} Figure \ref{fig:theta} evaluates the performance of our $Top\text{-}kCS^2$ algorithm for different similarity thresholds, $\theta$, where all other parameters are set to their default values. With the increase of threshold $\theta$, both wall clock time and I/O cost smoothly increase. This is because, a larger $\theta$ value will lead to a larger distance threshold (i.e., $k$-th largest distance from candidate community to $v_q$ in $Q$), which results in more candidates with higher processing cost. Nonetheless, for different $\theta$ values, both wall clock time and I/O cost remain low (i.e., $0.09\sim 0.29$ $sec$ and $189 \sim 341$ I/Os, respectively).

\begin{figure}[t!]
\centering 
\subfigure[][{\small wall clock time}]{
\scalebox{0.23}[0.23]{\includegraphics{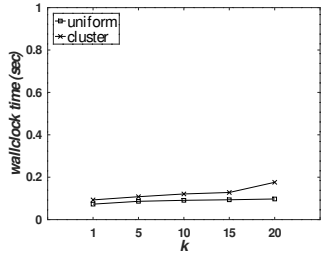}}
}
\subfigure[][{\small I/O cost}]{
\scalebox{0.23}[0.23]{\includegraphics{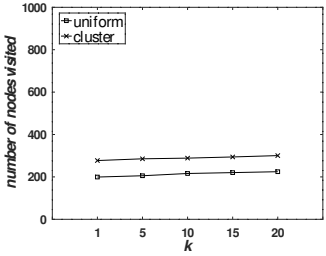}}
}
\caption{\small The $Top\text{-}kCS^2$ performance vs. parameter $k$.}
\label{fig:k}
\end{figure}

\begin{figure}[t!]
\centering 
\subfigure[][{\small wall clock time}]{
\scalebox{0.23}[0.23]{\includegraphics{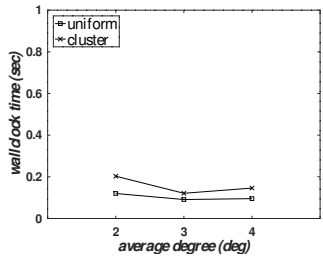}}\label{subfig:time_degree}
}
\subfigure[][{\small I/O cost}]{
\scalebox{0.23}[0.23]{\includegraphics{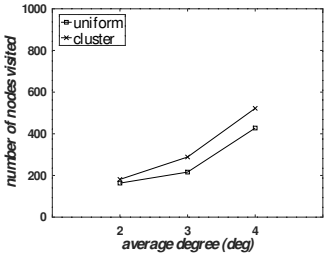}}\label{subfig:IO_degree}
}
\caption{\small The $Top\text{-}kCS^2$ performance vs. average degree $deg$.}\label{fig:avg_deg}
\end{figure}

\noindent \textbf{The $Top\text{-}kCS^2$ performance vs. graph size $|V(G)|$.} Figure \ref{fig:dataset} evaluates the scalability of our $Top\text{-}kCS^2$ algorithm by varying the number, $|V(G)|$, of the vertices from $10K$ to $100K$, where other parameters are set to their default values. Note that, the default value for the average degree of each vertex is 3, thus, there are about $300K$ edges for $100K$ vertices. Also, each edge has multiple POIs, making the number of POIs up to millions. From figures, when the graph size, $|V(G)|$, becomes larger, both the wall clock time and I/O cost of our $Top\text{-}kCS^2$ approach slightly increase. This is reasonable, since larger data sets lead to more candidate unit patterns (and candidate communities) to process and refine, which requires higher CPU time and node accesses in the aR-tree. Nonetheless, the wall clock time and I/O cost remain low (i.e., $0.1 \sim 0.4$ $sec$, and $83\sim 743$ I/Os, respectively), which shows good scalability of our $Top\text{-}kCS^2$ approach for different graph sizes $|V(G)|$.

\begin{figure}[t!]

\subfigure[][{\small wall clock time}]{
\scalebox{0.23}[0.23]{\includegraphics{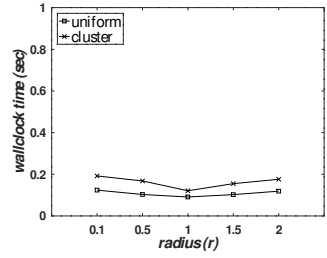}}
}
\subfigure[][{\small I/O cost}]{
\scalebox{0.23}[0.23]{\includegraphics{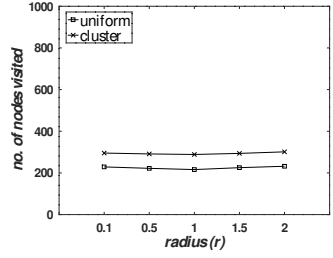}}
}
\caption{\small The $Top\text{-}kCS^2$ performance vs. radius $r$.}
\label{fig:radius}
\end{figure}

\begin{figure}[t!]
\centering 
\subfigure[][{\small wall clock time}]{
\scalebox{0.23}[0.23]{\includegraphics{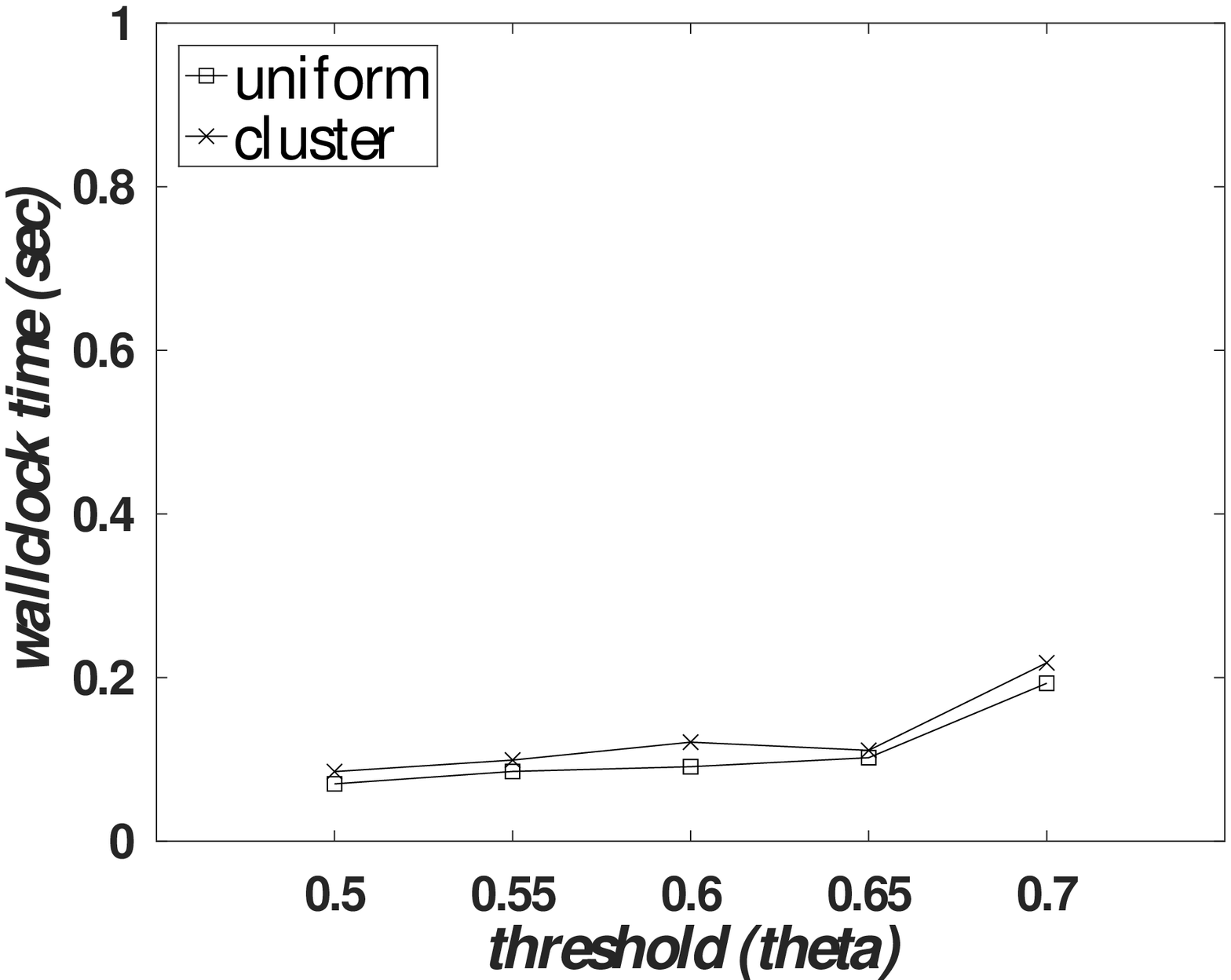}}
}
\subfigure[][{\small I/O cost}]{
\scalebox{0.23}[0.23]{\includegraphics{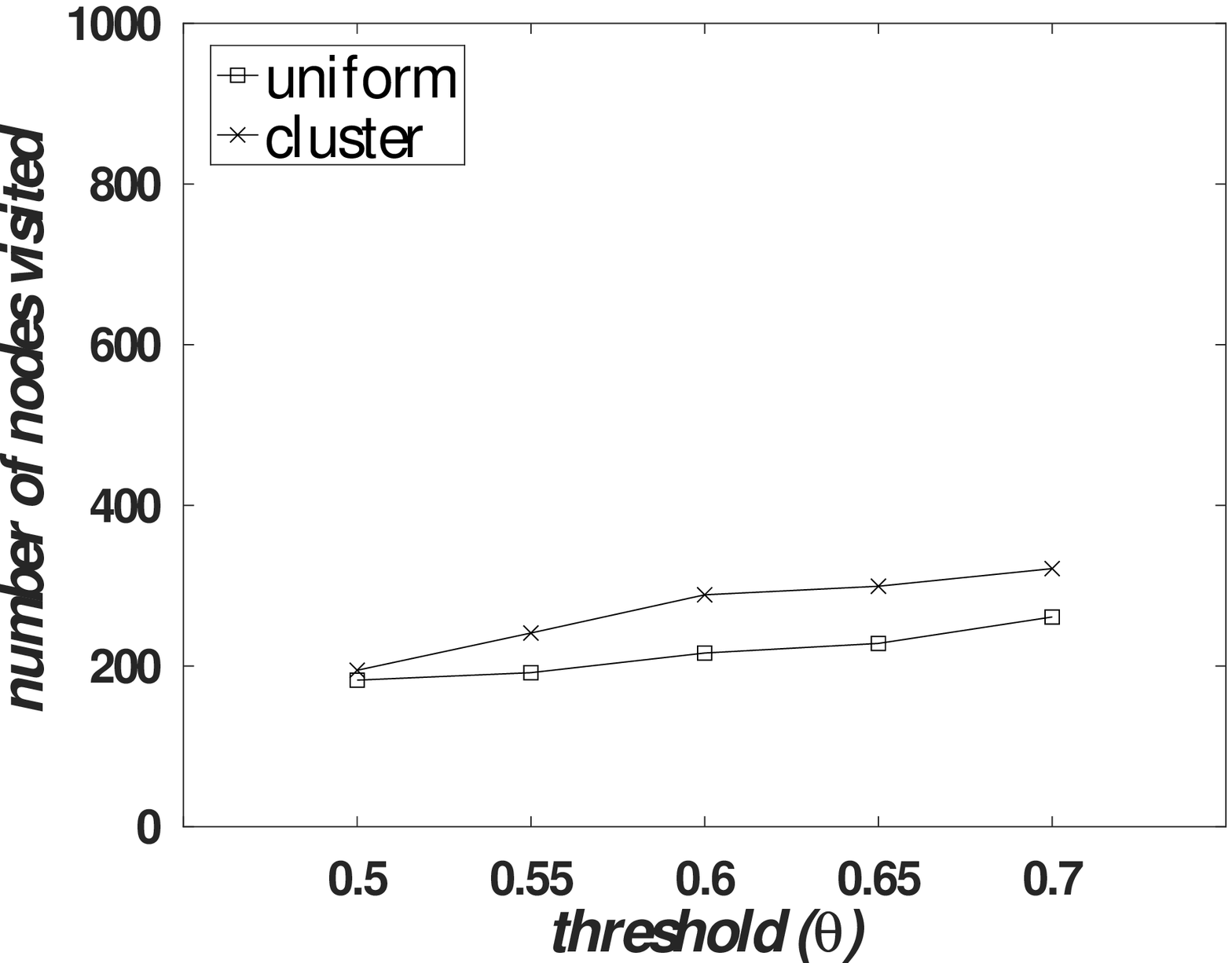}}
}
\caption{\small The $Top\text{-}kCS^2$ performance vs. similarity threshold $\theta$.}
\label{fig:theta}
\end{figure}

\begin{figure}[t!]
\centering 
\subfigure[][{\small wall clock time}]{
\scalebox{0.23}[0.23]{\includegraphics{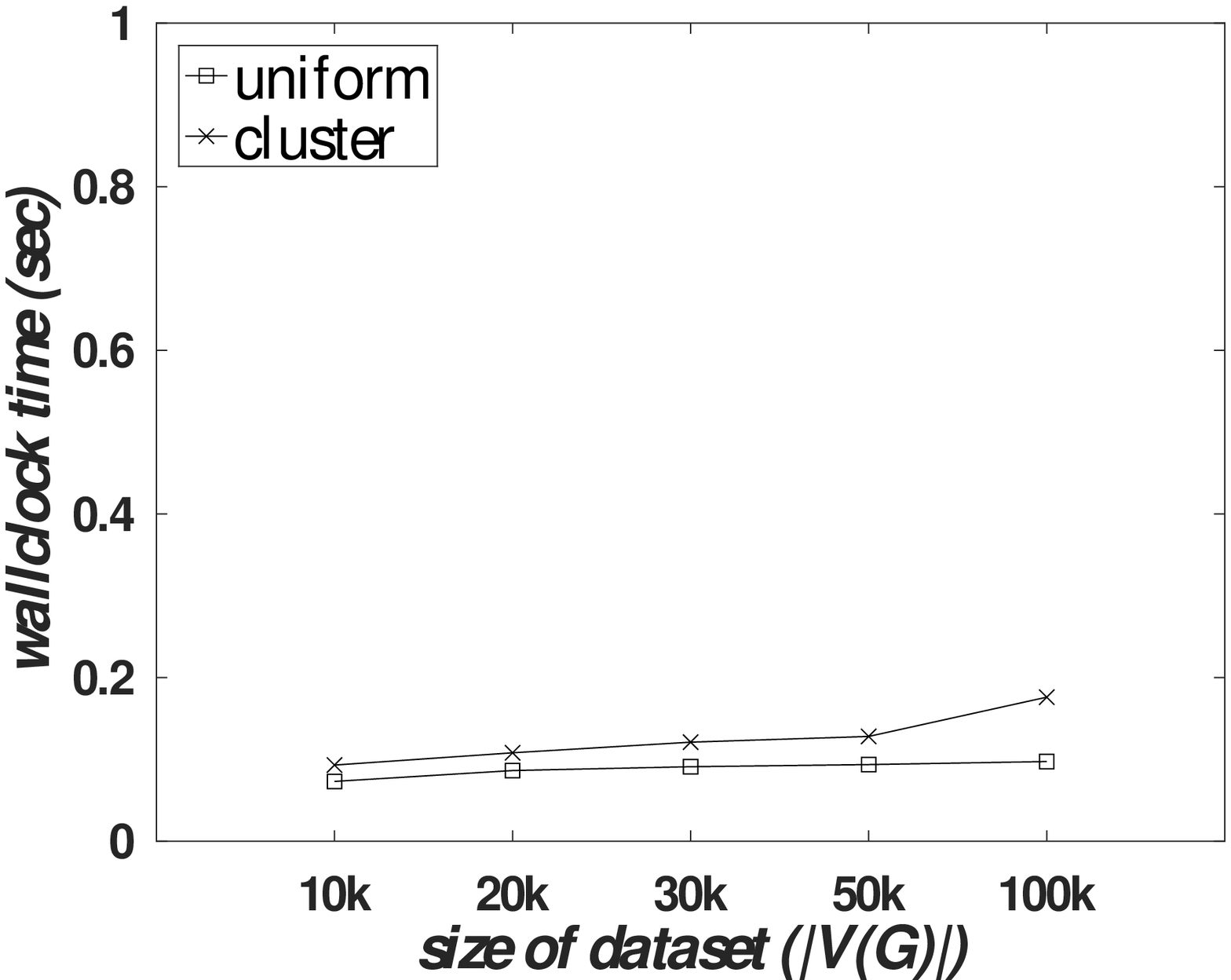}}\label{subfig:dataset_time}
}
\subfigure[][{\small I/O cost}]{
\scalebox{0.23}[0.23]{\includegraphics{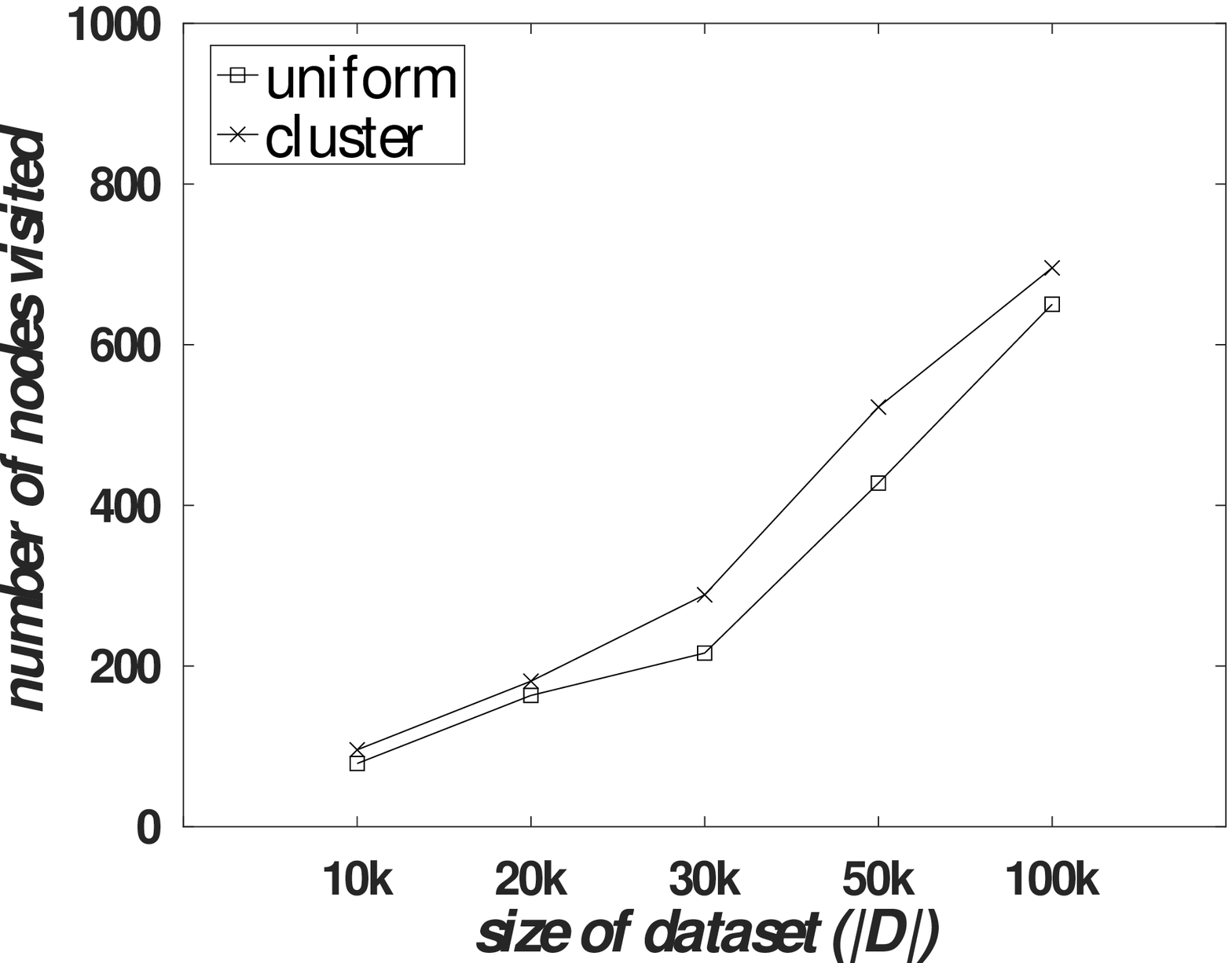}}\label{subfig:dataset_io}
}
\caption{\small The $Top\text{-}kCS^2$ performance vs. graph size $|V(G)|$.}
\label{fig:dataset}
\end{figure}

\subsection{Evaluation of the $CTop\text{-}kCS^2$ Query Efficiency}

\noindent {\bf The $CTop\text{-}kCS^2$ performance vs. real/synthetic data sets.} Figure \ref{fig:cont_baseline} compares our $CTop\text{-}kCS^2$ approach with the $baseline$ algorithm over real/synthetic data sets, in terms of the wall clock time. From the figure, we can see that the efficiency of the $CTop\text{-}kCS^2$ query outperforms that of $baseline$ for all the three data sets. This is because $CTop\text{-}kCS^2$ applies effective pruning methods with the help of the index. The experimental results confirm the effectiveness of our pruning methods, and efficiency of our $CTop\text{-}kCS^2$ approach. 

\noindent \textbf{The $CTop\text{-}kCS^2$ performance vs. parameter $k$.} Figure \ref{fig:cont_k} illustrates the effect of parameter $k$ on the $CTop\text{-}kCS^2$ query performance, where $k = 1, 5, 10, 15, 20$, and other parameters are set to default values. When we increase $k$, the wall clock time smoothly increases, whereas the I/O cost also slightly increases. Nonetheless, both wall clock time and I/O cost remain low (i.e., $0.49 \sim 0.87$ $sec$ and $235 \sim 315$ I/Os, respectively), which indicates the query efficiency of our proposed $CTop\text{-}kCS^2$ approach with different $k$ values.

\noindent \textbf{The $CTop\text{-}kCS^2$ performance vs. average degree $deg$.}
Figure \ref{fig:cont_avg_deg} varies the average degree, $deg$, of each vertex in the graph $G$ from 2 to 4, where other parameters are set to default values. When the $deg$ is small, there are many similar unit patterns (with simple structures), which leads to high processing cost. On the other hand, when the average degree, $deg$, becomes larger, more complex unit patterns (e.g., rectangles and triangles) need to be processed, which requires higher cost. Therefore, for larger average degree, $deg$, as shown in Figure \ref{subfig:cont_time_degree}, the wall clock time first decreases, and then increases. Regarding the I/O cost, since larger $deg$ indicates that we need to access more vertices/edges to search for candidate patterns, the I/O cost increases for larger $deg$ (as shown in Figure \ref{subfig:cont_IO_degree}). Nevertheless, the overall wall clock time and I/O cost remain low (i.e., less than $0.76$ $sec$ and less than 553 I/Os, respectively).

\noindent \textbf{The $CTop\text{-}kCS^2$ performance vs. radius $r$.} Figure \ref{fig:cont_radius} shows the wall clock time and I/O cost of our $CTop\text{-}kCS^2$ approach over $uniform$ and $cluster$ data sets, by varying the radius $r$ from 0.1 to 2, where default values are used for other parameters (as depicted in Table \ref{table:parameter_conti}). The experimental results show that the wall clock time and I/O cost are not very sensitive to radius $r$. Moreover, for different radius $r$, the wall clock time and I/O cost remain low (i.e., $0.6 \sim 0.77$ $sec$ and $250 \sim 313$ I/Os, respectively).

\begin{figure}[t!]
\centering 
\scalebox{0.15}[0.15]{\includegraphics{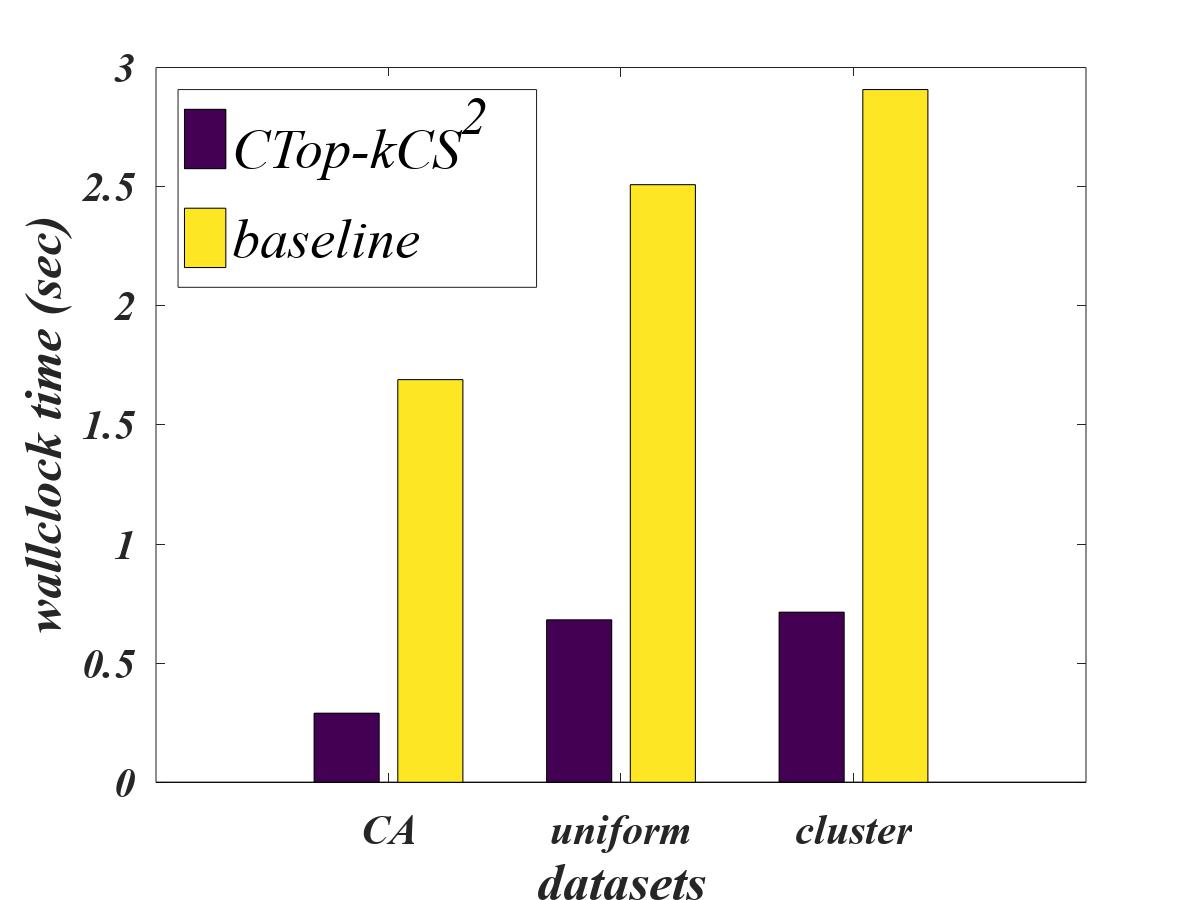}}
\caption{\small The $CTop\text{-}kCS^2$ performance vs. real/synthetic data.}
\label{fig:cont_baseline}
\end{figure}

\noindent \textbf{The $CTop\text{-}kCS^2$ performance vs. similarity threshold, $\theta$.} Figure \ref{fig:cont_theta} evaluates the performance of our $CTop\text{-}kCS^2$ algorithm for different similarity thresholds, $\theta$, where all other parameters are set to their default values. With the increase of threshold $\theta$, both wall clock time and I/O cost smoothly increase. This is because, a larger $\theta$ value will lead to a larger distance threshold (i.e., $k$-th largest distance from candidate community to $S_i$ in $Q_i$), which results in more candidates with higher processing cost. Nonetheless, for different $\theta$ values, both wall clock time and I/O cost remain low (i.e., $0.58\sim 0.83$ $sec$ and $220 \sim 361$ I/Os, respectively).

\noindent \textbf{The $CTop\text{-}kCS^2$ performance vs. graph size $|V(G)|$.} Figure \ref{fig:cont_dataset} evaluates the scalability of our $CTop\text{-}kCS^2$ algorithm by varying the number, $|V(G)|$, of the vertices from $10K$ to $100K$, where other parameters are set to their default values. Note that, the default value for the average degree of each vertex is 3, thus, there are about $300K$ edges for $100K$ vertices. Also, each edge has multiple POIs, making the number of POIs up to millions. From figures, when the graph size, $|V(G)|$, becomes larger, both the wall clock time and I/O cost of our $CTop\text{-}kCS^2$ approach slightly increase. This is reasonable, since larger data sets lead to more candidate unit patterns (and candidate communities) to process and refine, which requires higher CPU time and node accesses in the aR-tree. Nonetheless, the wall clock time and I/O cost remain low (i.e., $0.35 \sim 1.4$ $sec$, and $83\sim 755$ I/Os, respectively), which shows good scalability of our $CTop\text{-}kCS^2$ approach for different graph sizes $|V(G)|$.

\begin{figure}[t!]
\centering 
\subfigure[][{\small wall clock time}]{
\scalebox{0.23}[0.23]{\includegraphics{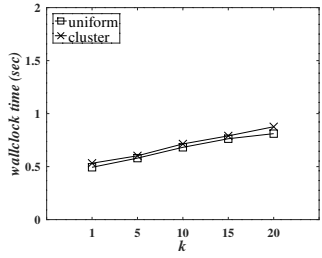}}
}
\subfigure[][{\small I/O cost}]{
\scalebox{0.23}[0.23]{\includegraphics{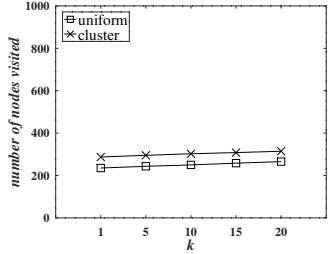}}
}
\caption{\small The $CTop\text{-}kCS^2$ performance vs. parameter $k$.}
\label{fig:cont_k}
\end{figure}

\begin{figure}[t!]
\centering 
\subfigure[][{\small wall clock time}]{
\scalebox{0.23}[0.23]{\includegraphics{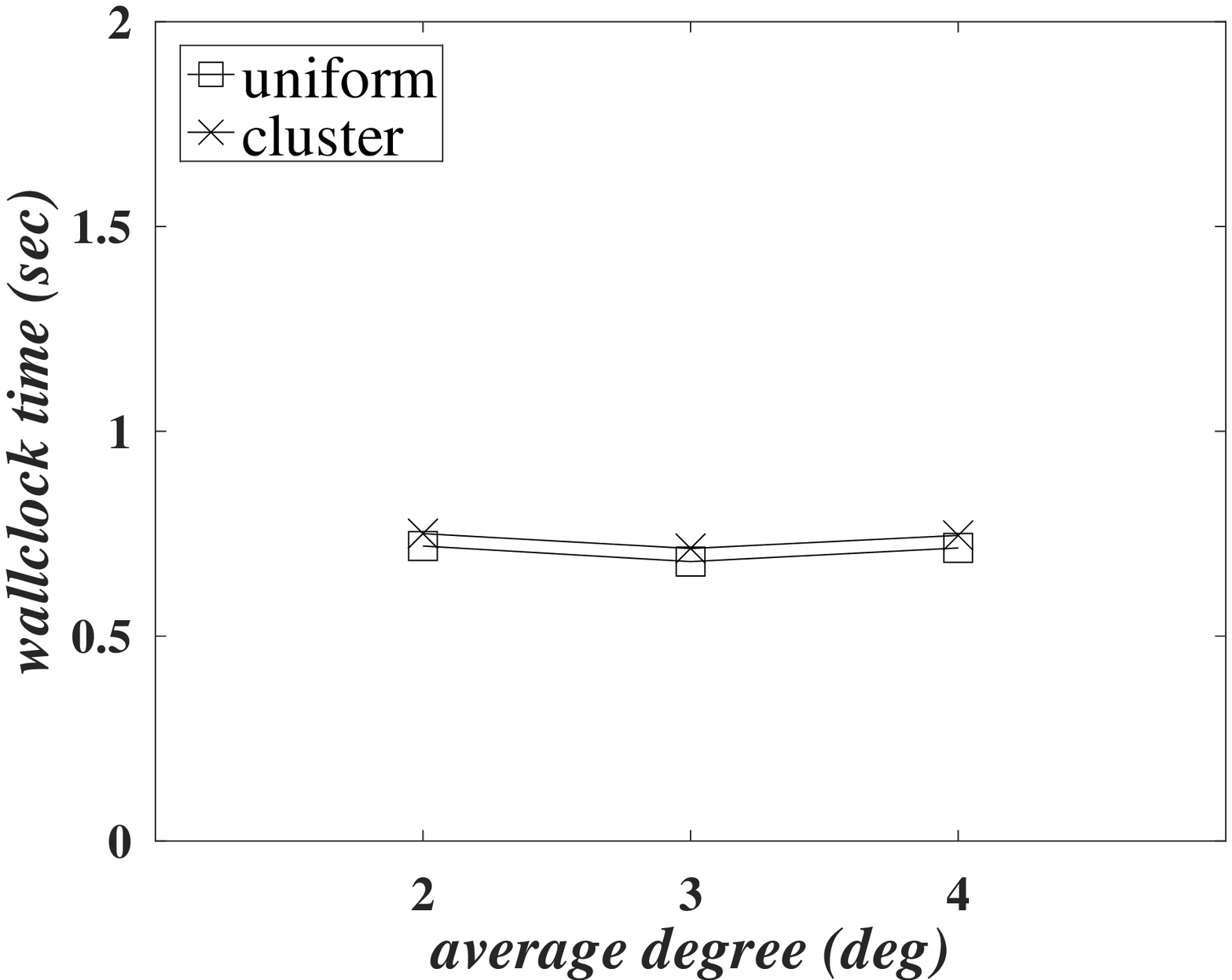}}\label{subfig:cont_time_degree}
}
\subfigure[][{\small I/O cost}]{
\scalebox{0.23}[0.23]{\includegraphics{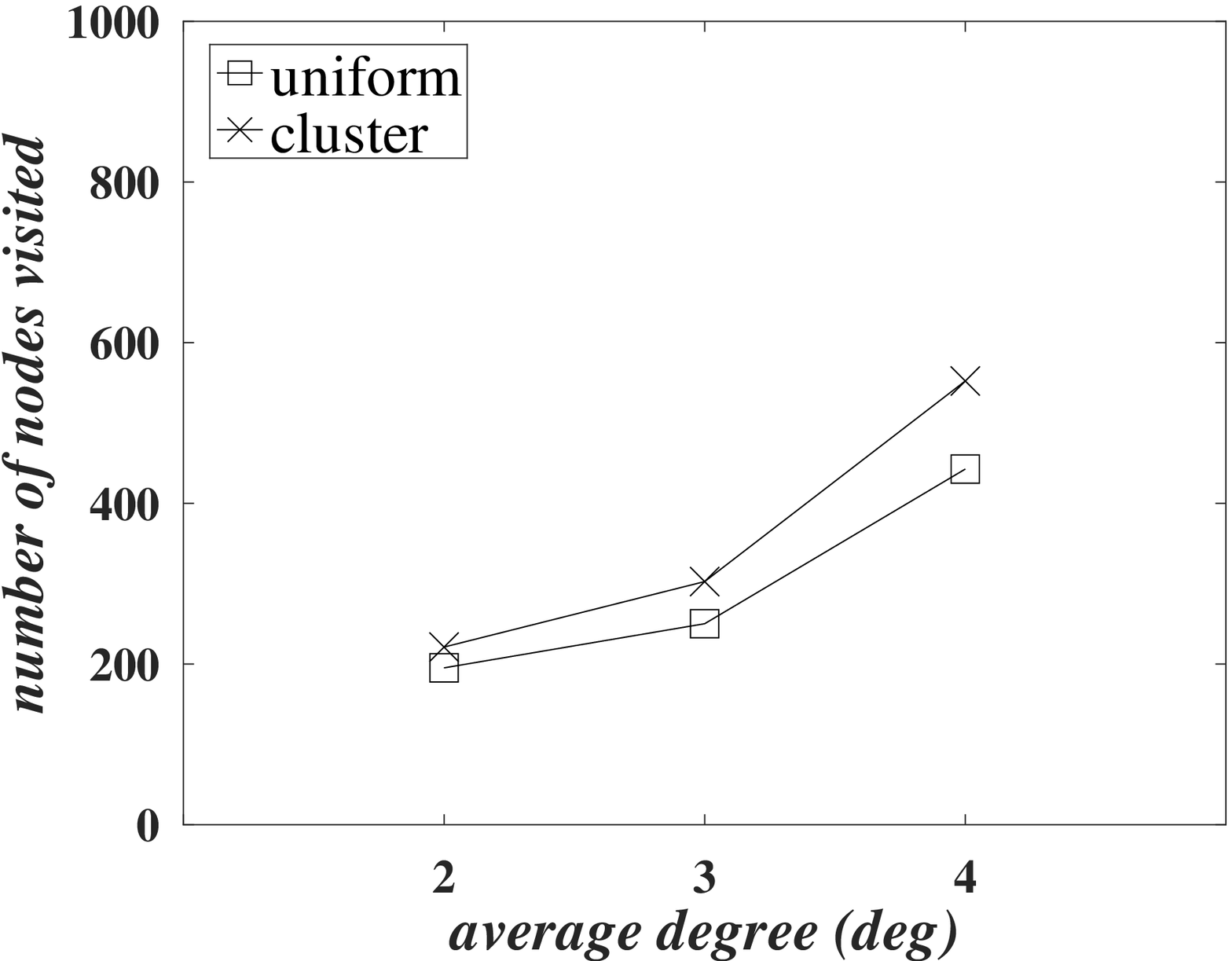}}\label{subfig:cont_IO_degree}
}
\caption{\small The $CTop\text{-}kCS^2$ performance vs. average degree $deg$.}\label{fig:cont_avg_deg}
\end{figure}

\begin{figure}[t!]
\centering 
\subfigure[][{\small wall clock time}]{
\scalebox{0.23}[0.23]{\includegraphics{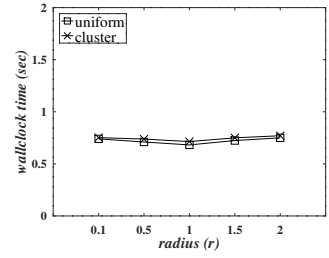}}
}
\subfigure[][{\small I/O cost}]{
\scalebox{0.23}[0.23]{\includegraphics{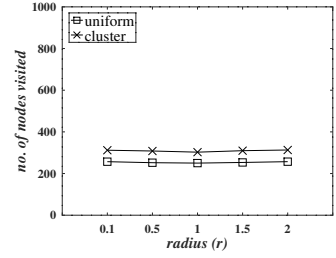}}
}
\caption{\small The $CTop\text{-}kCS^2$ performance vs. radius $r$.}
\label{fig:cont_radius}
\end{figure}

\noindent \textbf{The $CTop\text{-}kCS^2$ performance vs. query length $|L|$.} Figure \ref{fig:cont_query_length} evaluates the scalability of our $CTop\text{-}kCS^2$ algorithm by varying the size of the query line segment length, $|L|$, from $2$ to $6$, where other parameters are set to their default values. From figures, when the query length, $|L|$, increase, both the wall clock time and I/O cost of our $CTop\text{-}kCS^2$ approach slightly increase. This is reasonable, since larger query length lead to more query communities centered at $L$ which leads to more candidate unit patterns (and candidate communities) to process and refine, which requires higher CPU time and node accesses in the aR-tree. Nonetheless, the wall clock time and I/O cost remain low (i.e., $0.54 \sim 0.86$ $sec$, and $234 \sim 346$ I/Os, respectively).

\begin{figure}[t!]
\centering 
\subfigure[][{\small wall clock time}]{
\scalebox{0.23}[0.23]{\includegraphics{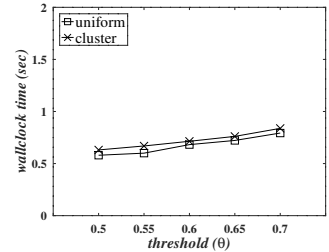}}
}
\subfigure[][{\small I/O cost}]{
\scalebox{0.23}[0.23]{\includegraphics{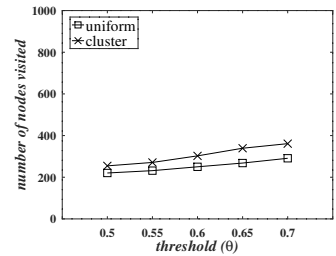}}
}
\caption{\small The $CTop\text{-}kCS^2$ performance vs. similarity threshold $\theta$.}
\label{fig:cont_theta}
\end{figure}

\begin{figure}[t!]
\centering 
\subfigure[][{\small wall clock time}]{
\scalebox{0.23}[0.23]{\includegraphics{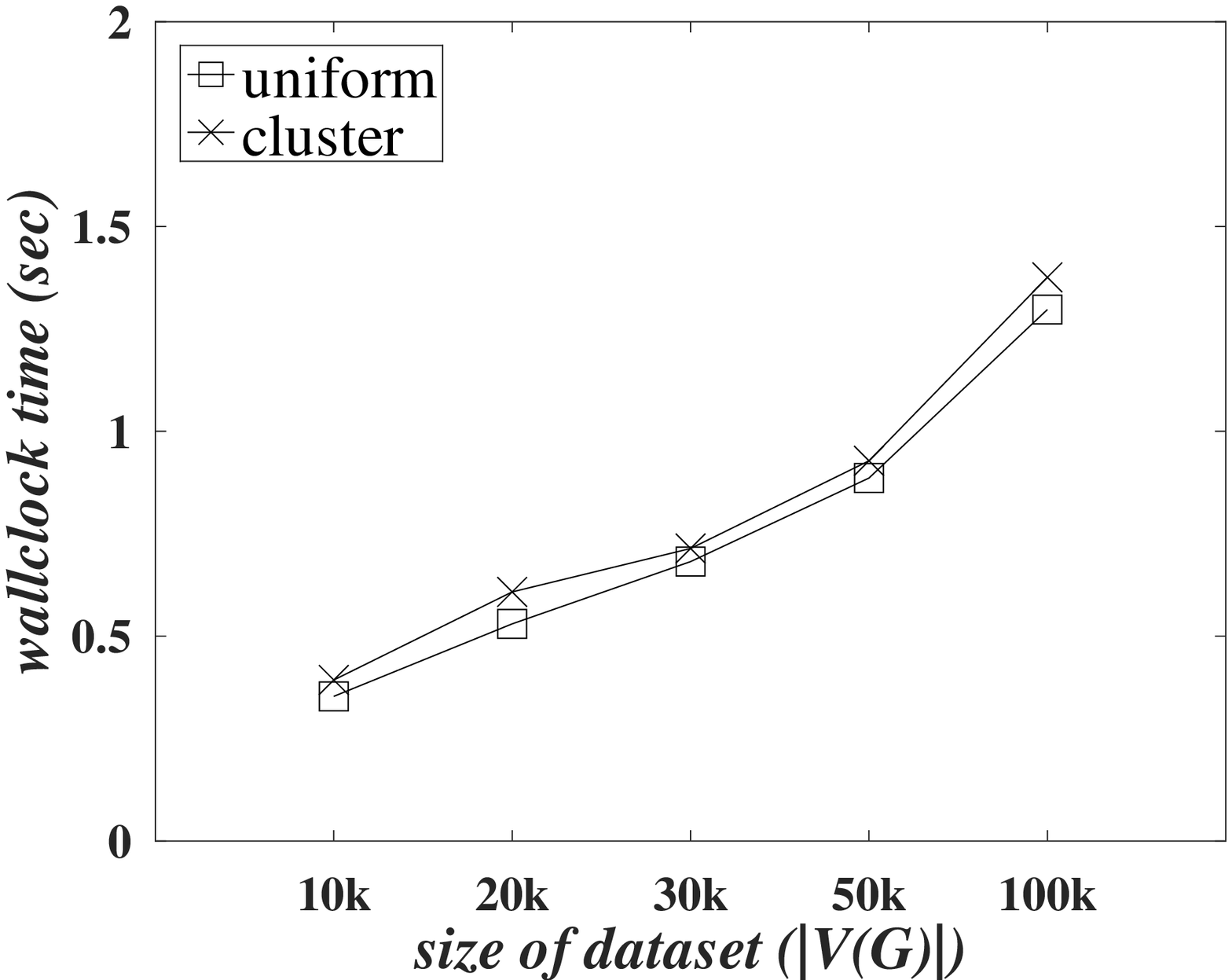}}\label{subfig:cont_dataset_time}
}
\subfigure[][{\small I/O cost}]{
\scalebox{0.23}[0.23]{\includegraphics{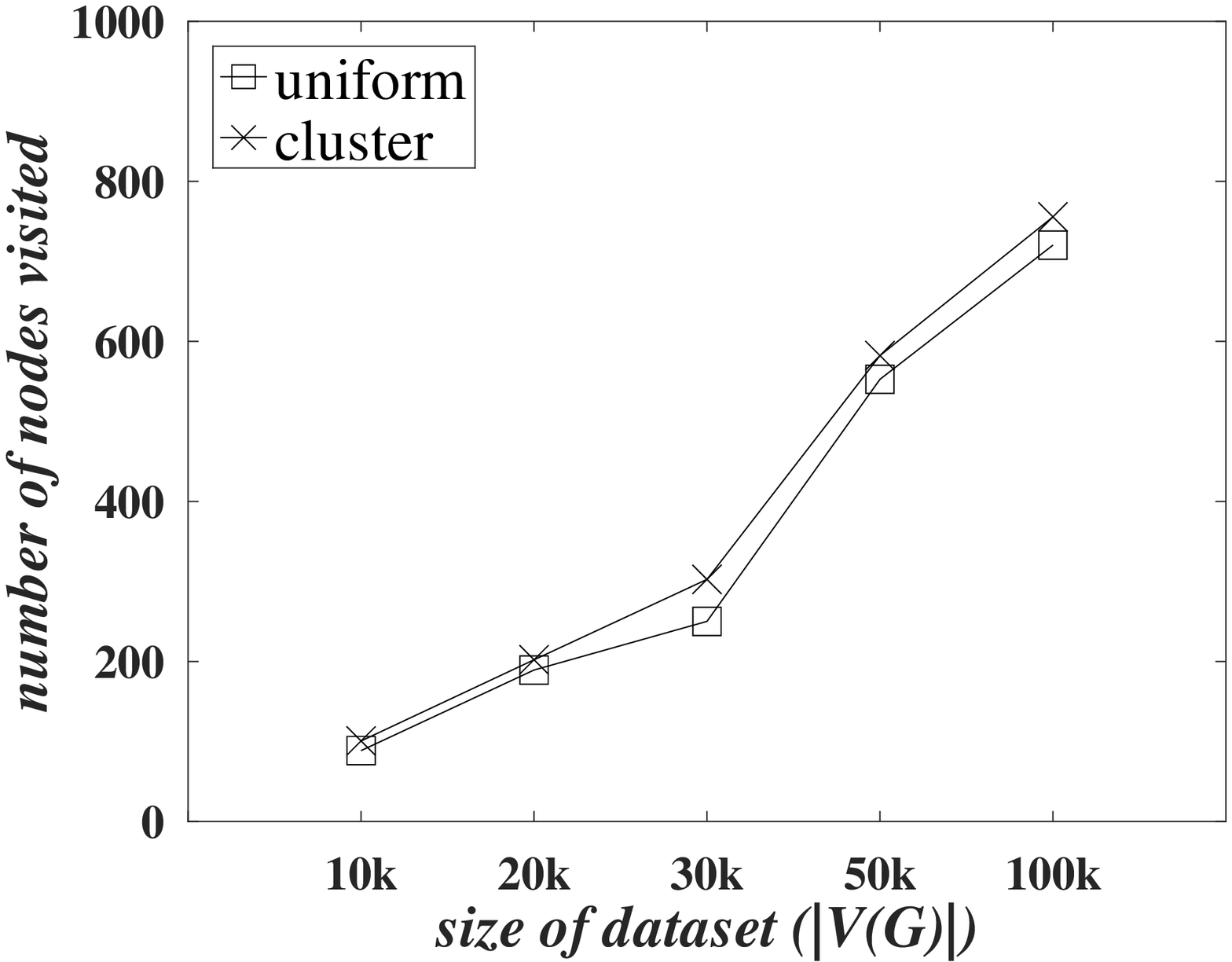}}\label{subfig:cont_dataset_cost}
}
\caption{\small The $CTop\text{-}kCS^2$ performance vs. dataset $|V(G)|$.}
\label{fig:cont_dataset}
\end{figure}

\begin{figure}[t!]
\centering 
\subfigure[][{\small wall clock time}]{
\scalebox{0.23}[0.23]{\includegraphics{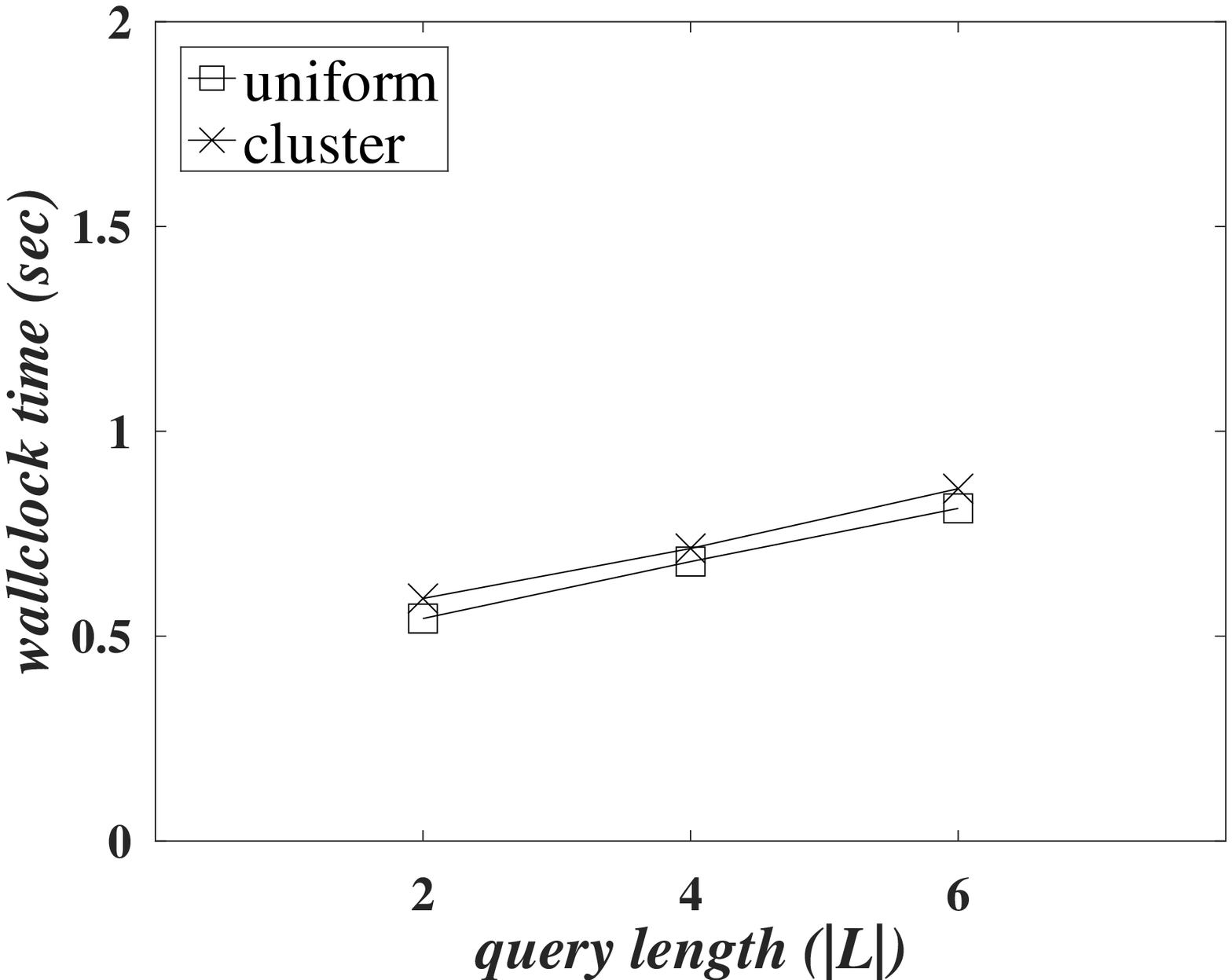}}\label{subfig:cont_L_time}
}
\subfigure[][{\small I/O cost}]{
\scalebox{0.23}[0.23]{\includegraphics{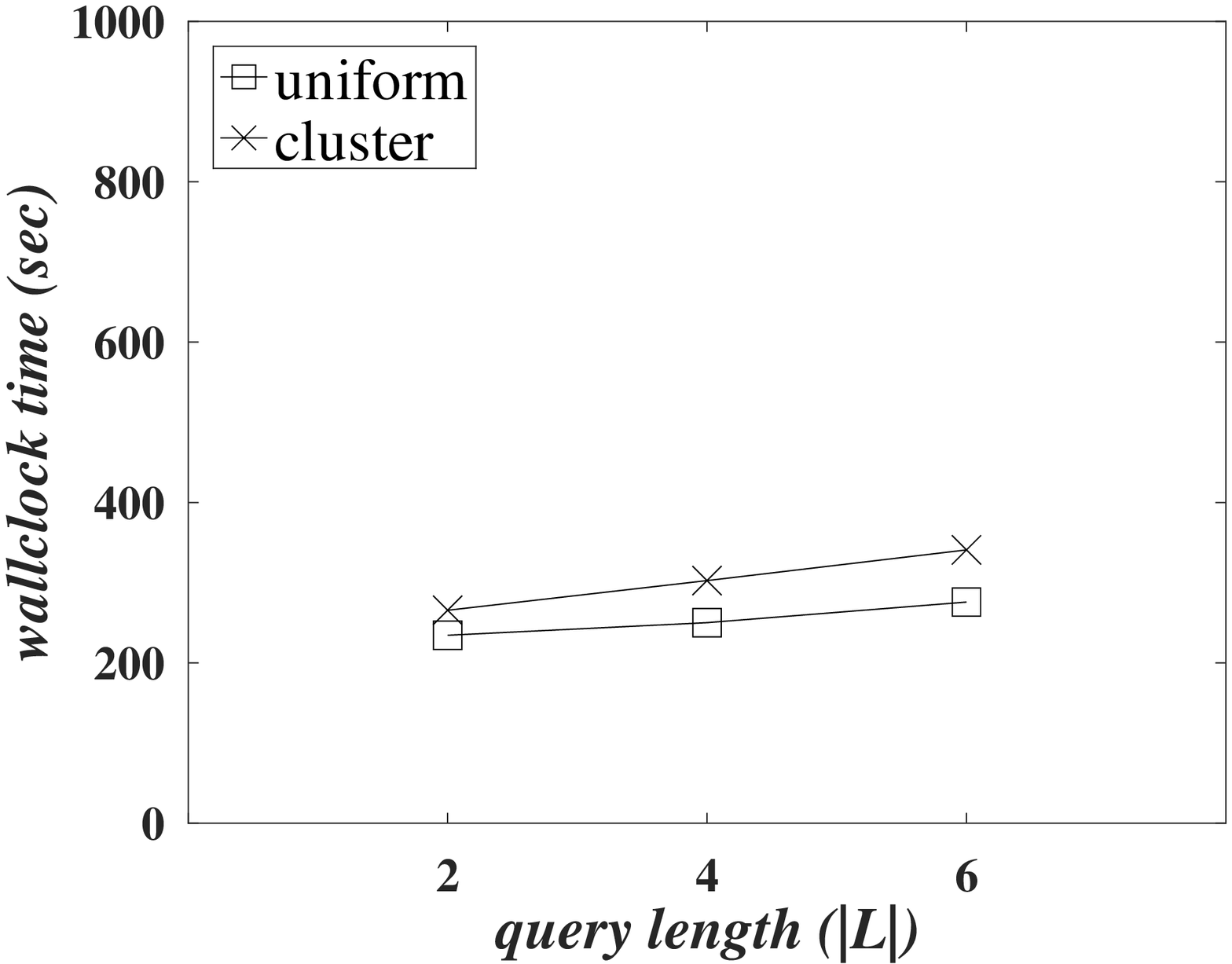}}\label{subfig:cont_L_io}
}
\caption{\small The $CTop\text{-}kCS^2$ performance vs. query length $|L|$. }
\label{fig:cont_query_length}
\end{figure}

\section{Related Work}
\label{sec:related_work}

In this section, we review related works on community detection and search in social networks.

\noindent \textbf{Community Detection.} In the last decade, discovering/detecting communities in social networks has been extensively studied in many real-world applications. Some prior works \cite{Fortunato10, Newman04} use link-based analysis to detect communities. Che et al. \cite{CheYW20memetic} discovered the community in signed networks. Conte et al. \cite{conte2020truly} used the $k$-truss and max-truss to detect communities in graphs with billions of edges. Wu et al. \cite{wu2020deep} applied the deep learning to detect high-quality communities in social networks. 

In contrast, our $Top\text{-}kCS^2$ problem is to retrieve top-$k$ communities (instead of all communities) from road networks (instead of social networks) that are similar and closest to query community $Q$. Thus, with different query semantics, we cannot borrow previous techniques for community detection to solve our problem.

\noindent \textbf{Community Search.} The community search problem aims to obtain communities that contain a given query vertex $q$. There are many existing works \cite{Girvan2002community, Expert2011uncovering, Clauset2004finding, guo2008region}, which proposed effective and efficient algorithms to retrieve such communities. Fang et al. \cite{fang2017effective} considered $k$-core semantics, which obtain communities (containing query vertex $q$) with the degree of each vertex greater than or equal to $k$. Conte et al. \cite{conte2020truly} studied the $k$-truss (i.e., a subgraph where each edge belongs to at least $(k-2)$ triangles) and $max$-truss (i.e., a $k$-truss community in graph $G$ with the maximum $k$ value). Sun et al. \cite{sun2020community} searched for communities with high Steiner connectivity. Li et al. \cite{li2015influential} defined $k$-influential community as a connected, cohesive subgraph with maximal structure.

There are some other works \cite{cui2013online, cui2014local, wu2015robust, huang2015approximate, li2015influential, sozio2010community, bronchi2015efficient, Clauset2005finding, Bagrow2008evaluating, vishwanath2010analysis} on searching communities (local neighborhood) over graphs. Specifically, Clauset \cite{Clauset2005finding} used ``local modularity'' as community goodness measure, which is the relative density within the community to outside the community. Bagrow et al. \cite{Bagrow2008evaluating} selected the largest ``outwardness'' of a vertex (i.e., the number of neighbors outside the community minus the number inside) to improve local community search. Viswanath et al. \cite{vishwanath2010analysis} observed that local community around the trusted node is also trustworthy. In \cite{yang2012socio, Li12, Yuan16, fang2017effective, GPSSN, chen2018maximum}, authors tried to find similar communities over geo-social networks with high structural and/or spatial cohesiveness. 

Different from these works on (geo-)social networks or attributed graphs, in this paper, we consider spatial communities (instead of user communities) on (planar) road networks with high similarities of graph structures (i.e., road-network patterns) and POIs (as well as distances to a query community $Q$). To our best knowledge, prior works did not study the community similarity search problem on road-network graphs, which takes into account factors such as graph patterns, POI similarity, and distances to a query community. Therefore, we cannot directly apply previous approaches for the community search to tackle our $Top\text{-}kCS^2$ problem.

Compared to our previous short conference paper \cite{rai2021top}, in this long version, for $Top\text{-}kCS^2$, we designed and provided more technical details in the offline pre-processing phase (Section \ref{sec:offline_processing}), pruning heuristics (Section \ref{sec:prune}), and candidate unit pattern retrieval (Section \ref{sec:candidate_unit}). We also reported the experimental results of our $Top\text{-}kCS^2$ algorithm for different parameter settings (Section \ref{sec:experiment}). Furthermore, we also formulate and tackle a new variant, that is, continuous $Top\text{-}kCS^2$ problem (denoted as $CTop\text{-}kCS^2$), where the query community moves along a query line segment (e.g., surrounding communities from home to the working place). We design novel techniques for $CTop\text{-}kCS^2$ to split the query line segment into multiple intervals (each with the same query community), and develop efficient algorithms to monitor and maintain top-$k$ communities that are similar to the query community for each interval.

\section{Conclusions}
\label{sec:conclusion}

In this paper, we formulate and tackle a novel problem of \textit{top-$k$ 
community similarity search} ($Top\text{-}kCS^2$) over large-scale road-network graphs, 
which retrieves $k$ spatial communities having high structural and POI similarities and with spatial closeness, 
with respect to a given query community. To tackle this problem, we propose effective pruning strategies 
and indexing mechanism, and develop an efficient $Top\text{-}kCS^2$ query processing algorithm.
We also consider and tackle the $CTop\text{-}kCS^2$ problem, where query communities continuously move along a line segment. We have demonstrated through extensive experiments the performance of our proposed 
$Top\text{-}kCS^2$ and $CTop\text{-}kCS^2$ approaches over real and synthetic road networks.

\balance



%\bibliographystyle{IEEEtran}

\bibliographystyle{IEEEtran}
%\bibliography{IEEEexample}
%\bibliographystyle{abbrv}
\let\xxx=\bibitem\def\bibitem{\par \xxx}
\bibliography{sample-bibliography}


\begin{IEEEbiography}[{\includegraphics[width=1in,height=1.25in,clip,keepaspectratio]{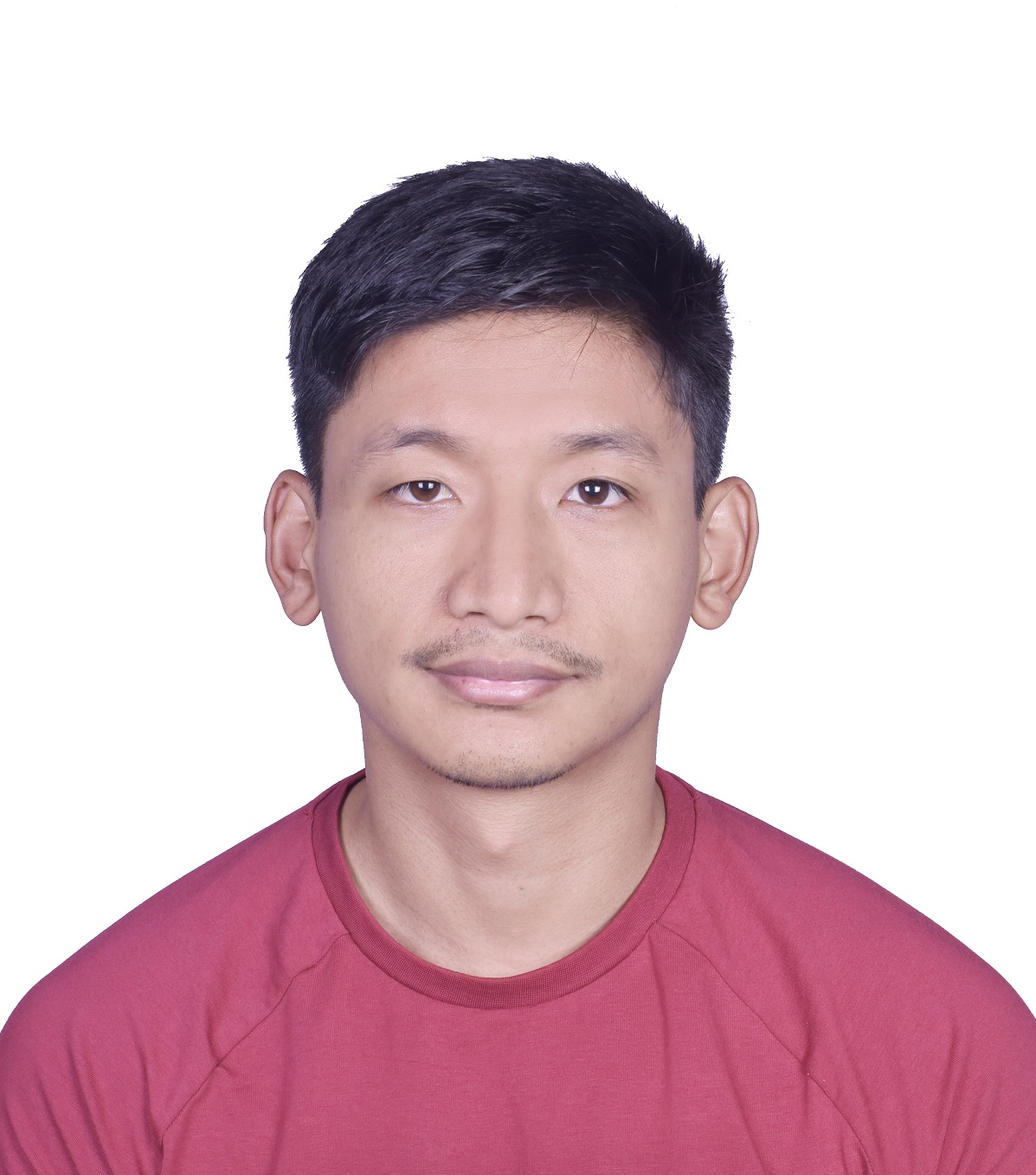}}]{Niranjan Rai} received the BE degree in computer engineering from Institute of Engineering, Tribhuvan University, Nepal, and is currently pursuing the PhD degree in computer science at Kent State University, USA. 
\end{IEEEbiography}

\begin{IEEEbiography}[{\includegraphics[width=1in,height=1.25in,clip,keepaspectratio]{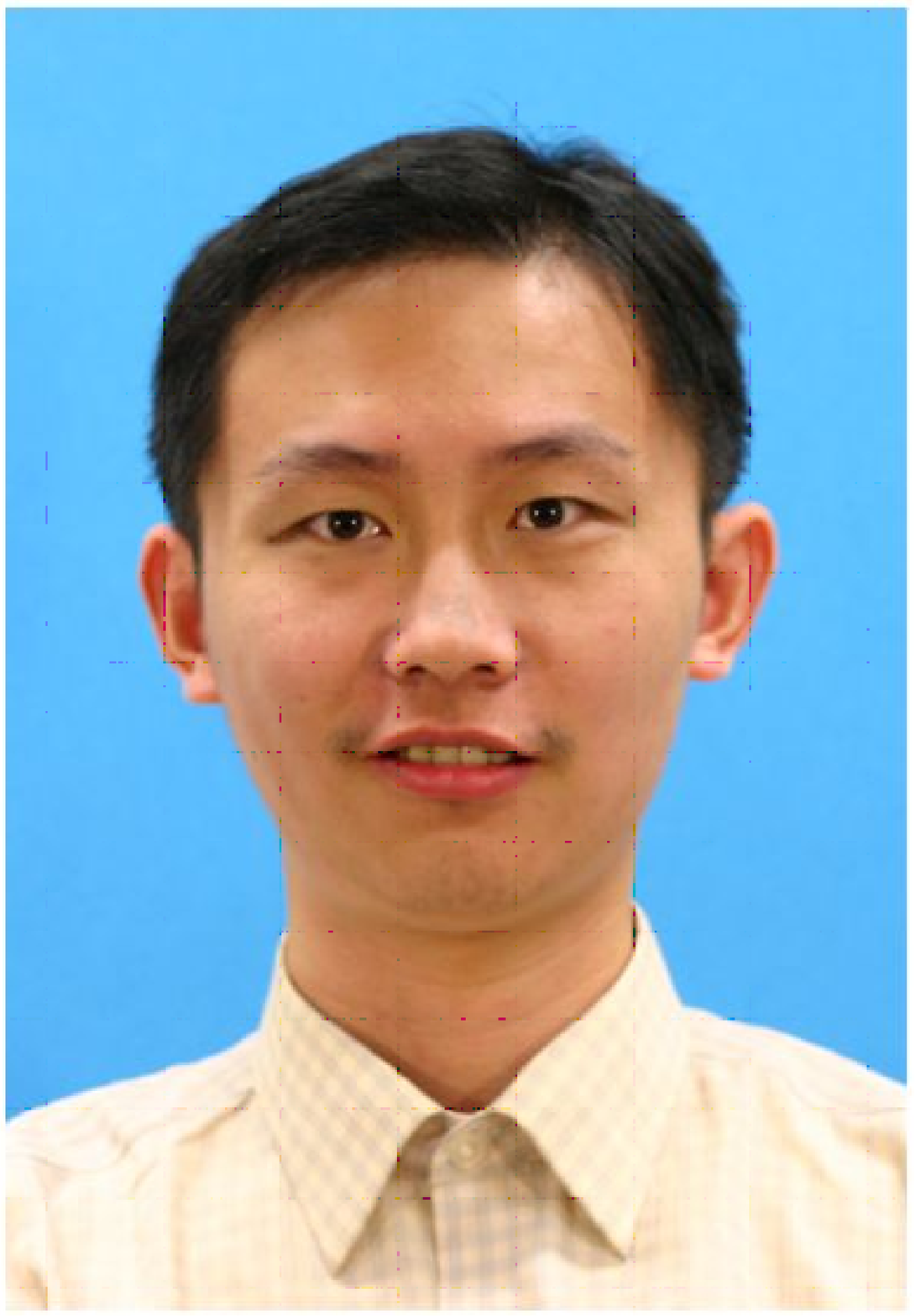}}]{Xiang Lian}
received the BS degree from the Department
of Computer Science and Technology,
Nanjing University, China, and the PhD degree in computer
science from the Hong Kong University of
Science and Technology, Hong Kong. He is now an associate
professor in the Computer Science Department
at Kent State University, USA. His research interests
include uncertain/certain graph databases, spatio-temporal databases, probabilistic databases, 
and so on.
\end{IEEEbiography}

\newpage

\appendices
\section{}
\noindent {\small{\bf Proof of Theorem \ref{theorem:sim_ub}}}

\begin{proof}
The score upper bound, $ub\_sim(C_l, Q)$, is the highest score possible for community, $C_l$ and $Q$. From Eqs.~(\ref{eq:community_score}) and (\ref{eq:ub_score}), we know that the difference between the similarity score and upper bound similarity score is the $POI$ vector used to calculate the score. Since, cosine similarity is the product of each component of two vectors, as shown in Eq.~(\ref{eq:cosine_sim_norm}), the larger the vector items, larger will be the score. We know from Figure \ref{fig:simscore}, max vector has larger vector components than the normal $POI$ vector in that unit pattern. The similarity score of a unit pattern increases when $max$ $POI$ vector, $c_h.max$,  is used to calculate the score. Ultimately, higher unit pattern scores leads to higher community similarity score. Therefore, $ub\_sim(C_l, Q)$ is always greater than $sim(C_l, Q)$. Hence, Theorem \ref{theorem:sim_ub} holds. 
\end{proof}

\section{}
\noindent {\small{\bf Proof of Theorem \ref{theorem:distancepruning}}}

\begin{proof}
From the assumption of the theorem, $C_k$ has the $k$-th largest distance, $dist(v_q, C_k)$, to the query community $Q$, and $dist(v_q, C_l) \geq dist(v_q, C_k)$ holds. Therefore, the community $C_l$ has the distance $dist(v_q, C_l)$ to $Q$ greater than or equal to at least $k$ candidate communities satisfying the similarity constraint $sim(C_i, Q) > \theta$  (i.e., $C_1$, $C_2$, ..., and $C_k$). This indicates that $C_l$ cannot be the $Top\text{-}kCS^2$ result, and thus can be safely pruned. Hence, Theorem \ref{theorem:distancepruning} holds. 
%Suppose we calculated communities $(C_1,~ C_2, \cdots,~ C_k)$ which satisfies the similarity threshold, $\theta$. According to Theorem \ref{theorem:distancepruning}, $dist(v_q, C_k)$ is the largest distance among $C_1 \sim C_k$. If there is a community $C_{k+1}$ which has similarity greater than threshold, $\theta$, i.e. $sim(C_{k+1}, Q) > \theta$. If $dist(v_q, C_{k+1})$ is greater than $dist(v_q, C_{k})$, it can never be the $top-k$, since community $C_k$ has the largest distance among $C_1$ to $C_k$. Hence, we can safely prune community $C_{k+1}$.
\end{proof}

\section{}
\noindent {\small{\bf Proof of Theorem \ref{theorem:unit_pattern_retrieval}}}
\begin{proof}
If a spatial community, $C_l$ is a candidate community, then the similarity score, $sim(C_l, Q)$ (given by Eq.~(\ref{eq:community_score})), between communities $C_l$ and $Q$ should be greater than or equal to $\theta$. According to the pigeonhole principle, for $n$ types of unit patterns, at least one unit pattern type in Eq.~(\ref{eq:community_score}) should have the similarity score greater than or equal to $\frac{\theta}{n}$ similarity score. In other words, we have: $\frac{\sum_{i=1}^{|c_h|}\sum_{j=1}^{|q_h|} cos\_sim(c_h[i], q_h[j])}{|q_h| \cdot n} \geq \frac{\theta}{n}$, which can be simplified to $\sum_{i=1}^{|c_h|}\sum_{j=1}^{|q_h|} cos\_sim(c_h[i], q_h[j]) \geq \theta\cdot |q_h|$. 

For $|c_h|$ unit patterns of type $h$, we can apply the pigeonhole principle again, and at least one unit pattern $c_h[i]$ has the summed similarity scores w.r.t. $q_h$ greater than or equal to $\frac{\theta\cdot |q_h|}{|c_h|}$. Thus, we obtain: $\sum_{j=1}^{|q_h|} cos\_sim(c_h[i], q_h[j]) \geq \frac{\theta\cdot |q_h|}{|c_h|}$. 

Since there are multiple possible communities in graph $G$ that contain $c_h[i]$, we can relax the threshold on the RHS of the inequality above to $\frac{\theta\cdot |q_h|}{\max\{|c_h|\}}$, where $\max\{|c_h|\}$ is an upper bound of the number of communities containing $c_h[i]$ (obtained via offline pre-computation). Therefore, if a unit pattern $c_h[i]$ satisfies $\sum_{j=1}^{|q_h|} cos\_sim(c_h[i], q_h[j]) \geq \frac{\theta\cdot |q_h|}{\max\{|c_h|\}}$ (which is exactly Eq.~(\ref{eq:eq6})), then $c_h[i]$ is a candidate unit pattern. Hence, the theorem holds.
\end{proof}

\ifCLASSOPTIONcaptionsoff
  \newpage
\fi

%\newpage


%\vfill


\end{document}